\documentclass{optica-article}

\journal{opticajournal} 

\articletype{Research Article}

\usepackage{graphicx}%
\usepackage{multirow}%
\usepackage{amsfonts}%
\usepackage{mathrsfs}%
\usepackage[title]{appendix}%
\usepackage{xcolor}%
\usepackage{textcomp}%
\usepackage{manyfoot}%
\usepackage{booktabs}%
\usepackage{algorithm}%
\usepackage{algorithmicx}%
\usepackage{algpseudocode}%
\usepackage{listings}%
\usepackage{siunitx}
\usepackage{url}

\def\HH{{\bold{ H}}}
\def\Z{{\mathbb{Z}}}
\newcommand{\bx}{\mathbf{x}}
\newcommand{\bs}{\mathbf{s}}
\newcommand{\bt}{\mathbf{t}}

\def\C{{\mathbb C}}
\def\Z{{\mathbb Z}}
\def\Y{{\mathbb Y}}
\def\R{{\mathbb R}}
\def\X{{\mathbb X}}

\def\DD{{ D}}
\def\CC{{\ C}}
\def\SS{\mathcal{S}}

\def\JJ{{I}_K}
\def\J{\mathcal{J}}
\def\K{\mathcal{K}}
\def\HH{{\mathcal H}}
\def\LL{{\mathcal L}}

\def\TT{{ T}}

\def\Re{{\mathsf{Re}}}

\newtheorem{theorem}{Theorem}%

\newenvironment{proof}{\noindent\textit{Proof:}}{\hfill$\square$}

\begin{document}

\title{{Efficient Near-Field Ptychography Reconstruction using the Hessian operator}}

\author{Marcus Carlsson\authormark{1,*}, Herwig Wendt\authormark{2}, Peter Cloetens\authormark{3}, Viktor Nikitin\authormark{4}}

\address{\authormark{1} Centre for Mathematical Sciences, Lund University, S\"olvegatan 18, 223 62 Lund, Sweden\\
\authormark{2} CNRS, IRIT, University of Toulouse, 2 rue Camichel,
31000 Toulouse, France\\
\authormark{3} ESRF-The European Synchrotron, 71 Avenue des Martyrs, 38043 Grenoble, France\\
\authormark{4} Advanced Photon Source, Argonne National Laboratory, 9700 S Cass Ave, Lemont, IL 60439, USA\\
}

\email{\authormark{*}marcus.carlsson@math.lu.se}

\begin{abstract*}
X-ray ptychography is a powerful and robust coherent imaging method providing access to the complex object and probe (illumination).
Ptychography reconstruction is typically performed using first-order methods due to their computational efficiency. Higher-order methods, while potentially more accurate, are often prohibitively expensive in terms of computation. In this study, we present a mathematical framework for reconstruction using second-order information, derived from an efficient computation of the bilinear Hessian and Hessian operator. The formulation is provided for Gaussian based models, enabling the simultaneous reconstruction of the object, probe, and object positions. Synthetic data tests, along with experimental near-field ptychography data processing, demonstrate a ten-fold reduction in computation time compared to first-order methods. The derived formulas for computing the Hessians, along with the strategies for incorporating them into optimization schemes, are well-structured and easily adaptable to various ptychography problem formulations.
\end{abstract*}

\section{Introduction}\label{sec1}

%
The increasing availability of coherent X-ray beams has significantly advanced the use of ptychography, a coherent diffraction imaging (CDI) technique that enables nano-resolution imaging without the need for traditional optical lenses. In ptychography, a partially coherent beam, referred to as the "probe", illuminates the sample that is shifted laterally through different positions, and the resulting diffraction patterns are recorded by a detector. The image is then reconstructed using computational phase retrieval methods. Depending on the experimental configuration, ptychography can be classified into near-field and far-field variants.

Near-field ptychography for nano-resolution imaging is typically performed with the sample positioned at a specific distance from an X-ray focal spot~\cite{stockmar2013near,clare2015characterization} such that a significant part of the sample is illuminated. A schematic of the setup is shown in Figure~\ref{fig:setup} a). The distances focus-to-sample and focus-to-detector define the geometrical magnification. The diffraction patterns at the detector plane are mathematically described by the squared magnitude of the Fresnel transform of the scattered (or exit) wave after propagation through the object.
In contrast, in far-field ptychography a much smaller part of the sample is illuminated by placing the sample at or near the focal spot~\cite{Rodenburg:08,rodenburg2019ptychography}, see Figure~\ref{fig:setup} b). The diffraction patterns can be approximated by the squared magnitude of the Fourier transform of the scattered (or exit) wave after propagation through the object. Due to the smaller sample illumination area in far-field ptychography, the number of scanning positions required to cover an equivalent sample area is typically much higher compared to near-field ptychography.

\begin{figure}[t]
\centering\includegraphics[width=.8\textwidth]{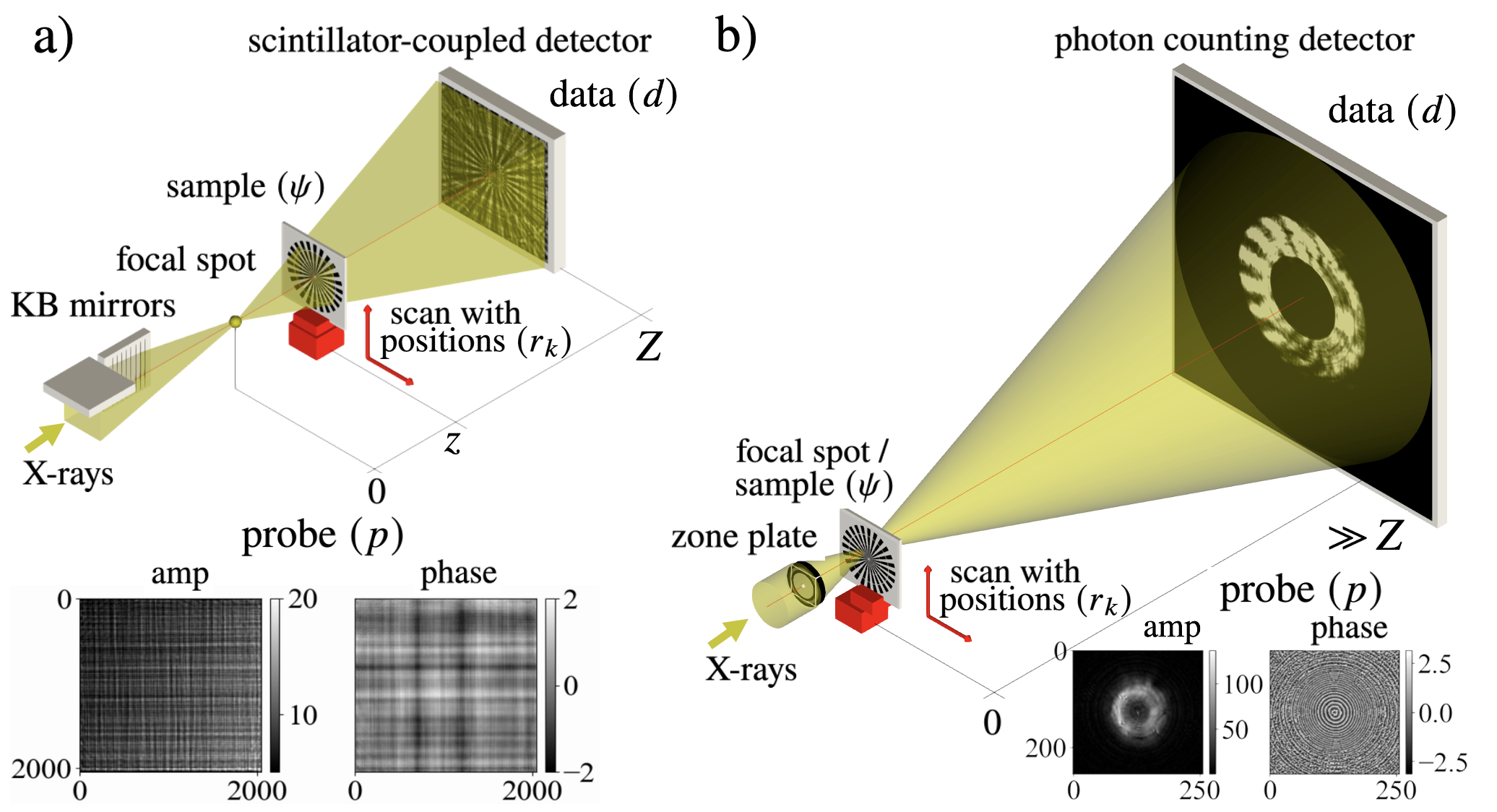}
\caption{
Experimental setups for (a) near-field ptychography, where diffraction patterns 
are recorded in the Fresnel regime and modeled using the Fresnel transform of 
the scattered/exit wave after propagation through the object, and (b) far-field 
ptychography, where diffraction patterns are recorded in the Fraunhofer regime
and modeled using the Fourier transform of the scattered/exit wave after 
propagation through the object. In both cases, beam focusing can be achieved 
using Kirkpatrick-Baez (KB) mirrors or a zone plate. The schematic setups also 
illustrate examples of the recovered probe amplitude and phase distributions.
}
\label{fig:setup}
\end{figure}
Both near-field and far-field ptychography techniques have found significant applications across a range of disciplines, including materials science, biology, and nanotechnology, enabling three-dimensional imaging at resolutions that surpass traditional microscopy~\cite{holler2017high,shahmoradian2017three,monaco2022comparison,aidukas2024high}.

In typical ptychography experiments, reconstruction is performed for both the unknown object $\psi$ and the probe $p$. Additionally, scan positions refinement is necessary to correct for potential motor errors.
Most modern ptychography algorithms rely on optimization of an objective functional $f$. While it is well known in optimization that knowledge of the second order derivatives of this functional, the so called Hessian matrix $\nabla^2f$ or $\boldsymbol{H}^f$,
can be used to speed up convergence significantly, this is usually not used in ptychography. This is partially due to the fact that the Hessian matrix is considered to have a complicated structure \cite{qian2014efficient}, but primarily because the amount of variables in $\psi$ that need to be reconstructed in a standard size modern application typically is of the order $10^4\times 10^4$. Thus, upon arranging these variables in a vector, the Hessian matrix would be of the size $10^8\times 10^8$, which is about 71 petabytes of complex numbers, each represented by two 32-bit floating point values.
This clearly makes it prohibitively slow or impossible to handle, and therefore first order methods and various types of alternating projection schemes are predominant, see \cite{chang2023fast} for a recent survey. Even for smaller scale experiments where $\psi$ is say $10^3\times 10^3$ and the probe is known, computing and storing the Hessian via standard calculus is not feasible. Using the Hessian matrix was tried, e.g.,~in \cite{yeh2015experimental}, reporting numerical results for $32\times32$-images only. 

The goal of this work is to introduce a new approach for efficiently using exact second order information in ptychography reconstruction, able to scale to large problem sizes such as those cited above. 
Our focus here is on describing the mathematical tools, rather than using this information to design or promote any particular algorithm, since ptychography problems come in a number of different settings, all with specialized algorithms designed for the particular issues arising in the given experimental setup. 
For this reason, and for the sake of brevity, we will only consider the case of near-field ptychography. However, our proposed approach can be equally well-suited for far-field ptychography, and we will explore this in future work.

To convey the key idea behind our method, assume for the moment that we have column-stacked our unknown variables so that they are represented by vectors $\bx$ in $\R^n$ (where $n$ is huge), and suppose that we
make use of an iterative reconstruction algorithm (i.e., a descent algorithm) and arrived at point $\bx_k$. 
\textit{If} we knew the gradient $\nabla f|_{\bx_k}$ and the Hessian matrix $\boldsymbol{H}^f|_{\bx_k}=\nabla^2f|_{\bx_k}$ of $f$ at $\bx_k$, we would have the following quadratic approximation 
\begin{equation}\label{secprx}
  f(\bx_k+\alpha \mathbf{s}_k)\approx f(\bx_k)+\alpha \langle \nabla f|_{\bx_k},  \mathbf{s}_k\rangle+\frac{\alpha^2}{2}\langle \boldsymbol{H}^f|_{x_k} \mathbf{s}_k,\mathbf{s}_k \rangle
\end{equation}
and could make use of it  in several ways in the process of determining the next point $\bx_{k+1}$:
 \begin{enumerate}
 \item An \emph{optimal step size} $\alpha_k$ known as Newton's rule could be determined given a descent direction $\mathbf{s}_k$ (e.g., the gradient $ \nabla f|_{\bx_k}$) by setting
 $
  \alpha_k= -\frac{\langle \nabla f|_{\bx_k},\mathbf{s}_k\rangle}{\langle \boldsymbol{H}^f|_{\bx_k}\mathbf{s}_k,\mathbf{s}_k \rangle}.
$
\item The \emph{conjugate gradient} (CG) method could be used with the principled conjugate directions known as Daniel's rule \cite{daniel1967conjugate}  (instead of the many heuristic rules that have been proposed to avoid computing the Hessian, see, e.g., the overview article~\cite{hager2006survey} and the recent book \cite{andrei2020nonlinear}), which requires computing the real number $\langle \boldsymbol{H}^f|_{\bx_k}\mathbf{s}_k,\mathbf{t}_k \rangle$ for two  vectors $\mathbf{s}_k$ and $\mathbf{t}_k$.
\item The \emph{Newton descent direction} $\mathbf{s}_k$ for second order (Quasi-)Newton algorithms could be obtained by (approximately) solving  for $\mathbf{s}_k$ the equation
\begin{equation}\label{eq2}
  -\nabla f|_{\bx_k}=\boldsymbol{H}^f|_{\bx_k} \mathbf{s}_k.
\end{equation}
 \end{enumerate}
 
%
In order to make use of these concepts in large scale problems, following the general ideas outlined in our theoretical paper \cite{Carlsson:25}, the first key observation is 
that none of the above expressions actually requires knowledge of the Hessian \emph{matrix}, but only computation of one real number $  \langle \boldsymbol{H}^f|_{\bx} \bs,\bt \rangle$, or one vector $\boldsymbol{H}^f|_{\bx} \bs$ in the case of the matrix-vector multiplication \eqref{eq2}. The second is that, once we have disposed of the matrix-vector multiplication, we can also dispose of the column-stacking (i.e.~the vectors) and represent data and the unknown variables \emph{in their natural domains}, i.e. 3D-tensors and matrices (2D-tensors), which further enables efficient and transparent implementation of the relevant operators in the reconstruction algorithm.
Our approach makes it possible to identify the \emph{analytic expressions} for the above quantities $\langle \boldsymbol{H}^f|_{\bx} \bs,\bt \rangle$ and $\boldsymbol{H}^f|_{\bx} \bs$ henceforth referred to as the \emph{bilinear Hessian} and the \emph{Hessian operator}, and to completely avoid the 
Wirtinger-type derivatives traditionally used for these type of computations~\cite{candes2015phase}. 
%

This work is organized as follows.  Section \ref{mat} introduces the mathematical framework for near-field ptychography reconstruction. Section \ref{sec:math} defines the bilinear Hessian and the Hessian operators as mathematical objects (Section \ref{ap1}), and we prove as a key contribution the chain rule for bilinear Hessians, allowing to decompose the computations into simple, transparent steps (Section \ref{secchainrule}). Section \ref{sec:optim} defines our proposed Gradient descent (GD), Conjugate Gradient (CG) and Quasi-Newton (QN) descent algorithms based on these operators.
Section \ref{sec:calc} derives the expressions for the gradient and Hessian for near-field ptychography.
In Section \ref{sec:results}, we present numerical results for simulated and experimental near-field ptychography data acquired at beamline ID16A of the European Synchrotron Radiation Facility (ESRF). Our proposed approach outperforms the state of the art in computational cost and accuracy. Section \ref{sec:conclusion} concludes this work and outlines future research.

Appendix I details the derivation of the expressions for blind ptychography with uncertain shift positions, and Appendix II shows how variable preconditioning for improving numerical stability can be seamlessly carried out implicitly within the proposed optimization algorithms.

\section{Mathematical framework for ptychography}\label{mat}

We are interested in reconstructing an unknown ``object'' $\psi$, which mathematically is represented as a matrix in $\C^{N\times N}$. For each scanning position $r_k=(x_k,y_k)\in \R^2$, where $k=1,\ldots,K$, a submatrix in $\C^{M\times M}$, which we denote by $\psi_{r_k}$, gets illuminated by a ``probe'' $p$ (also represented as a complex matrix $\C^{M\times M}$), and the joint wave-field $p\cdot \psi_{r_k}$ gets propagated to the detector by a diffraction operator $\DD$, such as the Fresnel transform (for near-field ptychography) or the Fourier transform (for far-field ptychography). Here `` $\cdot$ '' denotes elementwise multiplication of the matrices or tensors (a.k.a.~Hadamard multiplication). On the detector, only the intensity of the incoming field is measurable, i.e., for a given scanning position, each pixel value $(m_1,m_2)$ on the detector measures $$i_{m_1,m_2,k}=\Big|\big(\DD(p\cdot\psi_{r_k})\big)_{m_1,m_2}\Big|^2.$$
Thus data naturally comes arranged as a 3D-tensor $\R^{M\times M \times K}$, and the ptychography problem consists in retrieving the full matrix $\psi$, along with the probe $p$, from these measurements. The amplitude-based model provides more robust solutions than the intensity-based model, as concluded in \cite{yeh2015experimental,odstrvcil2018iterative}, so henceforth our ``data'' will be 
$d_{m_1,m_2,k}=\sqrt{i_{m_1,m_2,k}}$.

In summary, we wish to retrieve $(p,\psi)\in \C^{M\times M}\times\C^{N\times N}$ given data $d\in\R^{M\times M\times K}$. Eventually we will also consider the scanning positions as unknowns, but for the moment we ignore this for clarity of the exposition.
The forward model, in the absence of noise, can thus be written as 
\begin{equation}\label{forwardmodel}
  d=\Big|\DD_K\Big(\JJ(p)\cdot S_{{r}}(\psi)\Big)\Big|
\end{equation}
where $S_r$, $\JJ$ and $\DD_K$ are linear operators acting on the argument following:
\begin{itemize}
  \item $S_r$ is the operator which, for each scanning position $r_k$, ``extracts'' the illuminated part $\psi_{r_k}$ of $\psi$ and arranges these ``slices'' in a $\C^{M\times M\times K}$-tensor,
  \item $\JJ(p)$ denotes $K$ copies of $p$ arranged in the same $\C^{M\times M\times K}$-tensor,
  \item $\DD_K$ is the operator on $\C^{M\times M\times K}$ that takes each $\C^{M\times M}$-slice of the tensor and propagates it to the detector via $\DD$.\footnote{Note that when implementing in, e.g., Python, the operators $I_K$ and $D_K$ do not need to be considered explicitly since Python automatically distributes over the third variable when 3D tensors and matrices with compatible dimensions are multiplied, suggesting that they are also superfluous in the mathematical framework.
However, they are needed because the computation of the adjoint operators is also required, which involves summing operators over the third variable and needs to be coded explicitly.}
\end{itemize}  

In addition to \eqref{forwardmodel}, one needs to take into account noise. The noise is typically assumed to be Gaussian in near-field ptychography, where there is a relatively high photon count on the detector~\cite{odstrvcil2018iterative,nikitin2019photon}. 
As explained in e.g. \cite{odstrvcil2018iterative}, one wishes to minimize the negative log-likelihood of the joint probability distribution, which leads to the problem of finding the global minima of the objective functional
\begin{equation}\label{fG}
  f(p,\psi) =\sum_{\gamma \in \Gamma} w_\gamma \left(\Big|\DD_K\Big(\JJ(p)\cdot S_{{r}}(\psi)\Big)_\gamma\Big|-d_\gamma\right)^2,
\end{equation}
where $\Gamma$ denotes the index set $\{1,\ldots, N\}^2\times \{1,\ldots ,K\}$ and $w_\gamma$ is a weight which is zero at, e.g., possibly faulty detector positions, or reflects standard variance. 

\section{A general vectorization-free optimization framework}
\label{sec:math}

The functional \eqref{fG} acts on the linear space $\X=\C^{M\times M}\times \C^{N\times N}.$ On the other hand, most optimization methods are developed for $\R^n$, which is perfectly applicable in this setting since, upon ordering the elements and considering real and imaginary parts separately, we can identify $\X$ with $\R^n$ where $n=2M^2+2N^2$. This operation is usually done by column-stacking the matrices and then concatenating the resulting vectors. We will use the convention that any element $x\in\X$, when represented as a vector, is denoted in bold by $\bx$, and vice versa.
 
Thus, let us denote two arbitrary elements of the former space by $x$ and $y$ and their respective representations in $\R^n$ by $\bold{x}$ and $\bold{y}$. The \textit{bilinear Hessian} $\HH^f|_x$ of a given functional $f$, evaluated at a given point $x\in\X$, is the unique symmetric bilinear form on $\X$ such that we always have $$\HH^f|_{x}(u,v)=\langle \bold{H}^f|_{\bold{x}}\bold{u},\bold{v}\rangle, $$
where $\bold{H}^f|_\bold{x}$ is the Hessian matrix of $f$ when considered as a functional acting on $\R^n$ in the obvious way.
In \cite{Carlsson:25}, we presented a general technique for efficiently computing bilinear Hessians in large scale optimization settings such as those considered here. The core of our approach relies on Taylor expansions and Fr\'{e}chet-style derivatives in inner product spaces, extrapolating from classical works such as \cite{cartan1971differential,bauschke2017convex,absil2009optimization}. This allows us to avoid using the column-stacking operation, the classical chain rule and Wirtinger derivatives. 
Given the expression for $f$ from \eqref{fG},  we can derive expressions for $\HH^f|_x$ that do not explicitly rely on second order derivatives of $f$ and moreover are efficiently implementable. To this end, Theorem \ref{t1} presented in Section \ref{secchainrule} is also crucial.
Along the same lines, we introduced an efficiently computable forward operator $H^f|_x$ acting on $\X$, termed the \emph{Hessian operator}, with the property that $H^f_x(u)=v$ holds if and only if $\boldsymbol{H}^f|_{\bx}\bold{u}=\bold{v}$ holds. The mathematical details are outlined in the next section.

\subsection{Differential calculus in linear spaces}\label{ap1}

In this section we briefly recapitulate the main conclusions of \cite{Carlsson:25} and also establish a chain rule for computation of the bilinear Hessian.
Given two real inner product spaces $\X$ and $\Y$ and a function $\LL:\X\rightarrow \Y$, we define the Fr\'{e}chet derivative of $\LL$ at some $x_0\in\X$ (assuming that such exists) as the (unique) linear operator $\mathrm{d}\LL|_{x_0}:\X\rightarrow \Y$ such that $$\LL(x_0+\Delta x)=\LL(x_0)+\mathrm{d}\LL|_{x_0}(\Delta x)+o(\|\Delta x\|),$$   (where $o$ stands for ``little ordo''). Here and elsewhere, $\Delta x$ denotes an independent ``small'' variable in $\X$.
Similarly, we define the second order derivative as the (unique) bilinear symmetric $\Y$-valued operator such that 
\begin{equation}\label{expdef}
  \LL(x_0+\Delta x)=\LL(x_0)+\mathrm{d}\LL|_{x_0}(\Delta x)+\frac{1}{2}\mathrm{d}^2\LL|_{x_0}(\Delta x,\Delta x)+o(\|\Delta x\|^2).
\end{equation}
The fact that these objects are unique follows from Theorem 1 in \cite{Carlsson:25}, and will be used heavily in the derivations in Section \ref{sec:calc}.
Of course, if $\X=\R^n$ and $\Y=\R$, then the latter is simply the Hessian and the former becomes a linear functional and can be rewritten
$$\mathrm{d}\LL|_{x_0}(\Delta x)=\langle \nabla \LL|_{x_0},\Delta x\rangle,$$
for some element $\nabla \LL|_{x_0}\in \X$ that denotes the gradient. 
In this paper, the space $\X$ is typically a space of complex matrices or 3D-tensors. These spaces are naturally endowed with the complex inner product 
\begin{equation}\label{scalarcomplex}
\langle X,Y\rangle=\sum_{\gamma \in \Gamma}X_{\gamma}\overline{Y_{\gamma}},
\end{equation}
where $\Gamma$ denotes the index set, i.e.~$\{1,\ldots ,L\}\times \{1,\ldots ,M\}$ for $L\times M$-matrices and $\{1,\ldots ,L\}\times \{1,\ldots ,M\}\times\{1,\ldots ,N\}$ for $L\times M\times N$-tensors. However, in the proposed formulation, we can view these spaces as linear spaces over $\R$, where the new scalar product is given by \begin{equation}\label{scalarreal}
\langle X,Y\rangle_{\R}=\Re\langle X,Y\rangle,\end{equation} 
and then the techniques and results developed in this section apply all the same. Note in particular that \eqref{expdef} does not depend on whether the scalar product is real or imaginary, so the objects $\mathrm{d}\LL$ and $\mathrm{d}^2\LL$ are the same in both cases (and typically take complex numbers if $\Y$ is also a complex space).
However, if $\Y$ equals $\R$ then of course the values of $\mathrm{d}\LL$ are real and, to use this to obtain a \textit{gradient}, it is crucial to use the real scalar product. In this case, we will typically denote $\LL$ by $f$ and we define the gradient of $f$, denoted $\nabla f|_{x_0}$, as the unique element in $\X$ such that 
\begin{equation}\label{nabla}
    \mathrm{d}f|_{x_0}(\Delta x)=\langle \nabla f|_{x_0}, \Delta x\rangle_{\R}.
\end{equation}
Moreover $\mathrm{d}^2\LL$ is then called the \emph{bilinear Hessian} which we denote by $\HH^f$, and we define the \textit{Hessian operator} as the unique symmetric real linear operator $H^f:\X\rightarrow\X$ such that \begin{equation}\label{Hessianop}
\HH^f|_{x_0}(\Delta y,\Delta z)=\langle H^f|_{x_0}(\Delta y),\Delta z\rangle_\R    
\end{equation}
holds for all $\Delta y, ~\Delta z\in\X$. Armed with these definitions, we can now generalize \eqref{secprx} to a completely vector free setting as follows
\begin{equation}\label{secprxfree}
  f(x_0+ y)\approx f(x_0)+ \langle \nabla f|_{x_0},  y\rangle+\frac{1}{2} {\HH}^f|_{x_0} (y,y)
\end{equation}
Again, we refer to \cite{Carlsson:25} for a fuller discussion of these concepts and theoretical roots in the math literature. Expressions for $\nabla{f}|_{x_0}$ and ${\HH}^f|_{x_0}$ for the functional \eqref{fG} are derived in Sections \ref{sec:Objective} and \ref{man}, and the Hessian operator is derived in Section \ref{sec:HesOp}. 

The proposed approach bears equivalences with the use of Wirtinger derivatives in the sense that it generalizes the notion of gradients to functions on complex spaces. However, the proposed method relies on Fr\'echet derivatives to generalize this to multidimensional objects, which are more versatile than the standard derivatives used in Wirtinger calculus. This allows for the handling of more complicated functionals that are difficult to express using Wirtinger derivatives. Moreover, by deriving gradients through Taylor expansions rather than manipulating Wirtinger derivatives, the proposed approach offers simpler and more intuitive expressions. This generalization not only retains the familiar computational structure but also provides greater flexibility, making it particularly powerful for advanced applications in ptychography.

\subsection{The chain rule for bilinear Hessians}\label{secchainrule}

In order to keep the corresponding expressions as simple as possible and facilitate easy implementation, it is crucial to write the full objective functional as a composition of simpler functionals and then rely on a \textit{chain rule} for bilinear Hessians. 
Specifically, we now consider how to efficiently compute the above objects for composite functions $\J=\K\circ \LL$, where $\K:\Y\rightarrow \Z$ is another function between inner product spaces. 
We then have 
\begin{theorem}\label{t1}
Let $x_0\in\X$ be given and set $y_0=\LL(x_0)$. The joint derivative of $\J$ at $x_0$ is then given by 
\begin{equation}\label{e1}
    \mathrm{d}\J|_{x_0}(v)=\mathrm{d}\K|_{y_0} (\mathrm{d}\LL|_{x_0}(v)), \quad v\in \X.
\end{equation}
Moreover, the second order derivative is given by 
\begin{equation}\label{e2}
\begin{aligned}
  \mathrm{d}^2\J|_{x_0}(v,w)=\mathrm{d}^2\K|_{y_0}(\mathrm{d}\LL|_{x_0}(v),\mathrm{d}\LL|_{x_0}(w))+\mathrm{d}\K|_{y_0}(\mathrm{d}^2\LL|_{x_0}(v,w)).  
\end{aligned}
  \end{equation}
\end{theorem}

We remark that, in the case $\Z=\R$ and, denoting $\J$ by $f$ and $\K$ by $F$, the formula \eqref{e1} can be recast as 
\begin{equation}\label{nablacomp}
  \nabla f|_{x_0}=\mathrm{d}\LL|_{x_0}^*(\nabla F|_{y_0}),
\end{equation}
where $\cdot^*$ denotes the operator adjoint. To see this, note that in this case we have $$    \mathrm{d}f|_{x_0}(v)=\mathrm{d}F|_{y_0} (\mathrm{d}\LL|_{x_0}(v))=\langle \nabla F|_{y_0},\mathrm{d}\LL|_{x_0}(v)\rangle_{\R}= \langle \mathrm{d}\LL|_{x_0}^*(\nabla F|_{y_0}),v\rangle_{\R}.$$ Also note that this formula holds also if $\X$ and $\Y$ are complex spaces since, as noted after \eqref{scalarreal}, the adjoint is the same regardless of whether we use the complex or real scalar product. 

\begin{proof}
Setting $\Delta y=\LL(x_0+\Delta x)-\LL(x_0)$, we have $$\Delta y=\mathrm{d}\LL|_{x_0}(\Delta x)+\frac{1}{2}\mathrm{d}^2\LL|_{x_0}(\Delta x,\Delta x)+o(\|\Delta x\|^2).$$ By \eqref{expdef} applied to $\K$ at $y_0$ we thus get
\begin{align*}
  & \K(\LL(x_0+\Delta x))=\K(y_0)+ \mathrm{d}\K|_{y_0}(\Delta y)+\frac{1}{2}\mathrm{d}^2\K|_{y_0}(\Delta y,\Delta y)\Big\rangle+o(\|\Delta y\|^2)=\\& 
  \K(\LL(x_0))+ \mathrm{d}\K|_{y_0}\Big( \mathrm{d}\LL|_{x_0}(\Delta x)+\frac{1}{2}\mathrm{d}^2\LL|_{x_0}(\Delta x,\Delta x)\Big)+\\&\frac{1}{2}\mathrm{d}^2\K|_{y_0}\Big(\mathrm{d}\LL|_{x_0}(\Delta x),\mathrm{d}\LL|_{x_0}(\Delta x)\Big)+o(\|\Delta x\|^2),
\end{align*}
where the terms $\frac{1}{2}\mathrm{d}^2\LL|_{x_0}(\Delta x,\Delta x)$ are not present in the last line since, relying on the bilinearity of $\mathrm{d}^2\K|_{y_0}$, it is easy to see that these can be absorbed in the $o(\|\Delta x\|^2)$-term.
By the linearity of $\mathrm{d}\K|_{y_0}$ and the uniqueness of the operators in \eqref{expdef} (which is proved in Theorem 1 in \cite{Carlsson:25}), the formula \eqref{e1} is immediate, as well as the identity  
$$  \mathrm{d}^2\J|_{x_0}(\Delta x,\Delta x)=\mathrm{d}^2\K|_{y_0}(\mathrm{d}\LL|_{x_0}(\Delta x),\mathrm{d}\LL|_{x_0}(\Delta x))+\mathrm{d}\K|_{y_0}(\mathrm{d}^2\LL|_{x_0}(\Delta x,\Delta x)).$$
By this, \eqref{e2} follows by uniqueness of symmetric bilinear operators in expressions such as \eqref{expdef}, since indeed \eqref{e2} is a bilinear symmetric operator which agrees with the above expression on the diagonal $v=w=\Delta x$.
\end{proof}

\subsection{Optimization algorithms}
\label{sec:optim}

Before moving on to the concrete expressions that form the core of our contribution, let us briefly discuss how these can be used in various standard optimization methods such as Gradient Descent, Conjugate Gradient and Quasi-Newton algorithms with Newton step size for ptychography reconstruction. We term these algorithms BH-GD, BH-CG and BH-QN, respectively, where ``BH'' stands for ``bilinear Hessian''.
The expressions for $\nabla f|_{x}$, $\HH^f|_{x}$ and $H^f|_{x}$ are derived in the next Section and given in \eqref{nablaab}, \eqref{equ:biH} and Theorem \ref{t2}, respectively.

\subsubsection{Gradient descent and Newton step size}

Given a direction $s_k$, a near optimal strategy for step length for the step $x_{k+1}=x_k+\alpha_k s_k$ is to set $x_0=x_k$ and $y=\alpha s_k $ in the right hand side of \eqref{secprxfree} and minimize with respect to $\alpha$, which yields 
\begin{equation}\label{eqalpha}
               \alpha_k=-\frac{\langle \nabla f|_{x_k},s_k\rangle}{\HH^f|_{x_k}( {s_k},s_k)}.
             \end{equation}
In particular, if $s_k=-\nabla f|_{x_k}$, this gives the Gradient Descent method with Newton step-size;
\begin{equation}\label{GD}
  x_{k+1}=x_k-\frac{ \|\nabla f|_{x_k}\|^2}{\HH^f|_{x_k}(\nabla f|_{x_k},\nabla f|_{x_k})}\nabla f|_{x_k}.
\end{equation}

\subsubsection{Conjugate Gradient method}

The original CG method was developed as a fast solver to equations of the form $\bold{H} \mathbf{x}=\mathbf{d}$, where $\mathbf{H}$ is a positive definite matrix on $\R^n$ and $\mathbf{d}$ represents some measured data  \cite{hestenes1952methods}. 
It makes use of directions
$
\bold{s}_{k+1}=-\nabla f|_{\bx_{k+1}}+\beta_{k} \bold{s}_{k},
$
where $\beta_{k}$ is chosen so that $\langle \bold{s}_{k}, \bold{H} \bold{s}_{k-1}\rangle =0$, leading to the formula 
$
\beta_k=\langle \nabla f|_{\bx_{k+1}},\bold{H} \bold{s}_{k}\rangle/\langle \bold{s}_{k},\bold{H}\bold{s}_{k}\rangle.
$
For non-quadratic functionals, this expression for $\beta_k$ can be rewritten in many ways, leading to competing heuristic formulas such as  Fletcher-Reeves \cite{fletcher1964function}, Polak-Ribi\`{e}re \cite{polak1969note}, Hestenes-Stiefel \cite{hestenes1952methods}, Dai-Yuan \cite{dai1999nonlinear} and Hager-Zhang \cite{hager2005new}. 
Direct transposition of the expression for $\beta_k$ to the non-quadratic setting leads to
\begin{equation}\label{eqbeta}
\beta_k=\frac{\HH^f|_{x_{k+1}}(\nabla f|_{x_{k+1}},s_{k})}{\HH^f|_{x_{k+1}}( {s_{k}},s_{k})},
\end{equation}
which was proposed by Daniel \cite{daniel1967conjugate} but is rarely used because the computation of the Hessian is considered as a deadlock. Our expressions for the bilinear Hessian lift this deadlock, and we can define the following CG method with Newton step size:
\begin{align}
x_{k+1}&= x_k-\frac{\langle \nabla f|_{x_k},s_k\rangle}{\HH^f|_{x_k}( {s_k},s_k)} s_{k},\\
{s}_{k}&=-\nabla f|_{x_{k}}+\beta_{k-1} {s}_{k-1}.
\end{align}
\subsubsection{Quasi-Newton method}

Upon differentiating the right hand side of \eqref{secprxfree}, we see that its minimum is found at the points
\begin{equation}\label{QNeq1}
                       y=-H^f|_{x_k}^{-1}(\nabla f|_{x_k}),
\end{equation}
given that $H^f|_{x_k}^{-1}$ exists. But $H^f|_{x_k}^{-1}$ is intractable in the present context, whether or not it exists. Yet, assuming that it does, the equation can be solved approximately using, e.g., standard CG for a quadratic cost functional, see \cite[Section 3.2]{Carlsson:25} for details. This leads to an approximate Newton search direction $s_k\approx -H^f|_{x_k}^{-1}(\nabla f|_{x_k})$, and a Quasi-Newton update  $x_{k+1}=x_k+\alpha_k s_k$, with $\alpha_k$ given by \eqref{eqalpha}.

\section{Gradient, bilinear Hessian and Hessian operator}
\label{sec:calc}

\subsection{{Gradient, bilinear Hessian and Hessian operator on the detector}}
\label{sec:Objective}

We will rely on a nested scheme where the objective function $f$ in \eqref{fG} is written as a composition of a number of simpler functions, which are then composed to a full gradient and Hessian by the use of Theorem \ref{t1}. This strategy allows us to arrive at simple, interpretable and easy to implement expressions for the gradient, the bilinear Hessian and the Hessian operator.

Specifically, we introduce the functional $F:\C^{M\times M\times K}\rightarrow \R$ acting ``on the detector'',
\begin{equation}\label{FG}
  F(\Psi)=\sum_{\gamma \in \Gamma} w_\gamma\Big(|\Psi_\gamma|-d_\gamma\Big)^2,
\end{equation}
where the modulus is applied pointwise, and $\Psi=\DD_K\big[\JJ(p)\cdot S_{{r}}(\psi)\big]$ so that $$f(p, \psi)=F\Big(\DD_K\big[\JJ(p)\cdot S_{{r}}(\psi)\big]\Big).$$ Above, the meaning of $[\cdot]$ is the same as $(\cdot)$, and is introduced in order to make it easier to see which parentheses form pairs. 

As detailed in Section \ref{ap1}, $\C^{M\times M\times K}$ is a complex inner product space endowed with the usual inner product \eqref{scalarcomplex}, but we shall treat it as a linear space over $\R$ by introducing the real inner product \eqref{scalarreal}, which is crucial in order to be able to define gradients of the functional $f$.

To compute the gradient and bilinear Hessian of $F$, following the steps detailed in Procedure 1 from \cite{Carlsson:25}, we first use the Taylor expansion to derive 
\begin{equation}\label{modz}
    {|z_0+ \Delta z|}=|z_0|+\frac{\Re (\overline{z_0} \Delta z)}{|z_0|}+\frac{1}{2}\frac{  |\Delta z|^2}{|z_0|}-\frac{1}{2}\frac{\big[\Re (\overline{z_0} \Delta z)\big]^2}{|z_0|^3}+\mathcal{O}(| \Delta z|^3),
\end{equation}
where $z_0,~ \Delta z\in\C$ and $\mathcal{O}$ denotes ``big Ordo''. We then insert this in the expression \eqref{FG} to obtain
\begin{align*}
&F(\Psi_0+\Delta \Psi)=\\&\left\|\sqrt{w}\cdot\Big[|\Psi_0|-d+\Re \left(\frac{\overline{\Psi_0} \cdot {\Delta\Psi}}{|\Psi_0|}\right)+\frac{1}{2}\frac{|{\Delta\Psi}|^2}{|\Psi_0|}-\frac{1}{2}\frac{\big[\Re (\overline{\Psi_0}\cdot {\Delta\Psi})\big]^2}{|\Psi_0|^3}+\mathcal{O}\left(\| {\Delta\Psi}\|^3\right)\Big]\right\|^2,
\end{align*}
where all operations are applied elementwise and $w$ denotes the tensor $w=(w_\gamma)_{\gamma \in \Gamma}$. 
Expanding the inner product and collecting higher-order combinations in the $\mathcal{O}(\|{\Delta\Psi}\|^3)$-term we find that 
\begin{align*}
  &F(\Psi_0+{\Delta\Psi})=\Big\|\sqrt{w}\cdot\big(|\Psi_0|-d\big)\Big\|^2+2 \Re 
   \left\langle w\cdot\big(|\Psi_0|-d\big), \frac{\overline{\Psi_0}}{|\Psi_0|} \cdot {\Delta\Psi}\right\rangle+\\
  &\qquad \left\|\sqrt{w}\cdot{\Re \left(\frac{\overline{\Psi_0}\cdot {\Delta\Psi}}{|\Psi_0|}\right)}\right\|^2+ \left\langle {w}\cdot\big(|\Psi_0|-d\big), \frac{|{\Delta\Psi}|^2}{|\Psi_0|}-\frac{\big[\Re (\overline{\Psi_0}\cdot {\Delta\Psi})\big]^2}{|\Psi_0|^3}\right\rangle+\mathcal{O}\left(\|{\Delta\Psi}\|^3\right).
  \end{align*}
Due to some fortunate cancellations, it turns out that this can be simplified to
  \begin{align*}
&F(\Psi_0)+2 \Re 
   \left\langle {w}\cdot\big(\Psi_0-d\cdot \Psi_0/|\Psi_0|\big), {\Delta\Psi}\right\rangle+\\
  &\qquad  \left\langle w-w\cdot{d}/|\Psi_0|, {|{\Delta\Psi}|^2}\right\rangle+\left\langle  w\cdot{d}/|\Psi_0|^3,\big[\Re (\overline{\Psi_0}\cdot {\Delta\Psi})\big]^2\right\rangle+\mathcal{O}\left(\|{\Delta\Psi}\|^3\right)
\end{align*}
from which it immediately follows that 
\begin{equation}
    \nabla F|_{\Psi_0}=2 w\cdot \left(\Psi_0- d\cdot \Psi_0/|\Psi_0|\right).
\end{equation}
Similarly, by the uniqueness of the Hessian we see that $$\frac{1}{2}\HH^{F}|_{\Psi_0}({\Delta\Psi},{\Delta\Psi})=\left\langle w-w\cdot {d}/|\Psi_0|, {|{\Delta\Psi}|^2}\right\rangle+\left\langle {w}\cdot{d}/|\Psi_0|^3,\big[\Re (\overline{\Psi_0}\cdot {\Delta\Psi})\big]^2\right\rangle.$$
To obtain the Hessian as a symmetric bilinear form, one could use the polarization identity \cite[Eq. (14)]{Carlsson:25}, but by simply looking at the expression it is clear that \begin{align*}
  &\frac{1}{2}\HH^{F}|_{\Psi_0}({\Delta\Psi}^{(1)},{\Delta\Psi}^{(2)})=\left\langle w-w\cdot {d}/|\Psi_0|, \Re({{\Delta\Psi^{(1)}}}\cdot\overline{{\Delta\Psi^{(2)}}})\right\rangle+ \\&\qquad\left\langle w\cdot {d}/|\Psi_0|^3,\Re (\overline{\Psi_0}\cdot {\Delta\Psi}^{(1)})\cdot \Re (\overline{\Psi_0}\cdot {\Delta\Psi}^{(2)})\right\rangle
     \end{align*} 
is a symmetric bilinear form which coincides with the above on the ``diagonal'' $\Delta \Psi^{(1)}=\Delta \Psi^{(2)}=\Delta \Psi$, so by the uniqueness of such forms, this must be the sought expression \cite[Theorem 1]{Carlsson:25}. 

To arrive at the corresponding Hessian operator $H^{F}|_{\Psi_0}$, which acts on $\C^{M\times M \times K}$, we first note that the bilinear Hessian can be written
\begin{align}\label{HbiFG}
  \nonumber&\frac{1}{2}\HH^{F}|_{\Psi_0}\left({\Delta\Psi}^{(1)},{\Delta\Psi}^{(2)}\right)=\Re \left\langle (w-w\cdot {d}/|\Psi_0|)\cdot{{\Delta\Psi^{(1)}}},\Delta\Psi^{(2)}\right\rangle+ \\&\qquad\Re\left\langle w\cdot {d}\cdot \Psi_0/|\Psi_0|^3\cdot\Re (\overline{\Psi_0}\cdot {\Delta\Psi}^{(1)}),  {\Delta\Psi}^{(2)}\right\rangle
     \end{align} 
from which it immediately follows that 
\begin{equation}\label{HFG}
   H^{F}|_{\Psi_0}(\Delta \Psi)=2\left(w-w\cdot \frac{d}{|\Psi_0|}\right)\cdot{{\Delta\Psi}}+ 2w\cdot {d}\cdot \frac{\Psi_0}{|\Psi_0|^3}\cdot\Re \left(\overline{\Psi_0}\cdot {\Delta\Psi}\right). \end{equation}

\subsection{ The gradient and the bilinear Hessian for $f$}\label{man}

In this section we return to the functional $f$ as defined in \eqref{fG}, and derive the bilinear Hessian along with all gradients.
In order to do this using the abstract framework from Section \ref{ap1}, we note that $(p,\psi)$ naturally become variables in the space $\X=\C^{M\times M}\oplus\C^{N\times N}$, where $\oplus$ denotes the direct sum, i.e.~the linear space $\C^{M\times M}\times\C^{N\times N}$ endowed with the ``natural'' scalar product $$\langle (p_1,\psi_1),(p_2,\psi_2) \rangle_\X=\langle p_1,p_2\rangle_{\C^{M\times M}}+\langle \psi_1,\psi_2\rangle_{\C^{N\times N}}.$$
This is precisely the scalar product one would get if we were to column-stack and concatenate 
$(p,\psi)$ and identify it with an element of $\C^n$ with $n=M^2+N^2$. One could also consider the positions $r$ as variables, and in practice this is necessary for good reconstructions, but for simplicity of exposition we add this layer of complexity in Appendix I.

To clarify computations let us introduce the new variables $a=I_K(p)$ and $b=S_r(\psi)$.
We introduce the auxiliary function $\LL:(\C^{M\times M\times K})^2\rightarrow \C^{M\times M\times K}$ defined as $\LL(a,b)=D_K(a\cdot b)$ so that 
\begin{equation}\label{tildef}
     f(p,\psi)=F\Big(\LL\big[I_K(p),S_r(\psi)\big]\Big),
\end{equation}
where $(\C^{M\times M\times K})^2$ is short for $\C^{M\times M\times K}\oplus \C^{M\times M\times K}$.
By linearity of $D_K$ it follows that $$\LL(a_0+\Delta a,b_0+\Delta b)=\LL(a_0,b_0)+D_K\big(\Delta a \cdot b_0+a_0\cdot \Delta b+\Delta a\cdot \Delta  b\big).$$
The two middle terms give that  $\mathrm{d}\LL|_{(a_0,b_0)}(\Delta a, \Delta b)=
D_K(\Delta a\cdot b_0)+ D_K(a_0\cdot \Delta b)
$
whereas the latter yields
\begin{align*}
 & \mathrm{d}^2\LL|_{(a_0,b_0)}\big((\Delta a, \Delta b),(\Delta a, \Delta b)\big)=2 D_K( \Delta a\cdot \Delta b)
\end{align*}
where the factor 2 comes from the fact that $\mathrm{d}^2\LL$ in \eqref{expdef} has a factor $\frac{1}{2}$ in front of it.
By the uniqueness part of Theorem 1 in \cite{Carlsson:25}, the bilinear version of $\mathrm{d}^2\LL$ must thus be 
\begin{align*}
 & \mathrm{d}^2\LL|_{(a_0,b_0)}\big((\Delta a^{(1)}, \Delta b^{(1)}),(\Delta a^{(2)}, \Delta b^{(2)})\big)=D_K(\Delta a^{(1)}\cdot \Delta b^{(2)}+\Delta a^{(2)}\cdot \Delta b^{(1)})
\end{align*}
since this is a symmetric bilinear form which coincides with the former expression on the ``diagonal'' $\Delta a^{(1)}=\Delta a^{(2)}$, $\Delta b^{(1)}= \Delta b^{(2)}$.

Finally, we want to compose everything to get the gradient and bilinear Hessian for $f$. Setting $a_0=I_K(p_0)$, $\Delta a=I_K(\Delta p)$, $b_0=S_r(\psi_0)$ and $\Delta b=S_r(\Delta \psi)$ we get, by formulas \eqref{nabla} and \eqref{e1}, that $\Re \langle \nabla  f|_{(p_0,\psi_0)}, (\Delta p,\Delta \psi)\rangle_{\C^{M\times M}\oplus \C^{N\times N}}$ equals
\begin{align*}
  &\Re \Big\langle \nabla F|_{\LL(a_0,b_0)},\mathrm{d}\LL|_{(a_0,b_0)}\big[I_K(\Delta p),S_r(\Delta \psi)\big]\Big\rangle_{\C^{M\times M\times K}}=\\&\Re \Big\langle {\DD_K}^*(\nabla F|_{\LL(a_0,b_0)}), I_K(\Delta p)\cdot b_0+ a_0\cdot S_r(\Delta \psi)\Big\rangle_{\C^{M\times M\times K}},
\end{align*}
where we explicitly indicate which scalar product is referred to for extra clarity. Here, $\DD_K^*$ denotes the operator adjoint of ${\DD_K}$ and is simply the standard adjoint $\DD^*$ (i.e.~the back-propagation the detector to sample plane) applied individually to each of the $K$ data slices in $\C^{M\times M \times K}$. To shorten notation, we introduce \begin{equation}\label{Phi}
                                                    \Phi_0=\DD_K^*(\nabla F|_{\LL(a_0,b_0)}).
                                                                               \end{equation}
Summing up we see that \begin{align*}
  &\Re \langle \nabla  f|_{(p_0,\psi_0)}, (\Delta p,\Delta \psi)\rangle_{\C^{M\times M}\oplus \C^{N\times N}}=\\
  &\Re \big\langle I_K^*(\overline{b_0}\cdot\Phi_0),\Delta p\big\rangle_{\C^{M\times M}}+\Re \big\langle S_r^*(\overline{a_0}\cdot \Phi_0),\Delta \psi\big\rangle_{\C^{N\times N}}=\\&   \Re \left\langle \Big(I_K^*\big[\overline{S_r(\psi_0)}\cdot \Phi_0\big],S_r^*\big[\overline{I_K(p_0)}\cdot \Phi_0\big]\Big),(\Delta p,\Delta \psi)\right\rangle_{\C^{M\times M}\oplus \C^{N\times N}}
\end{align*}
where $I_K^*$, as noted in Section \ref{mat}, is simply an operator summing over the third index.
By definition \eqref{nabla} (with $\X=\C^{M\times M}\oplus \C^{N\times N}$), we see that
\begin{equation}\label{nablaab}
    \nabla f|_{(p_0,\psi_0)}= \Big(I_K^*\big[\overline{S_r(\psi_0)}\cdot\Phi_0\big],S_r^*\big[\overline{I_K(p_0)}\cdot\Phi_0\big]\Big)
\end{equation}
or equivalently that    $ \nabla_{p} f|_{(p_0,\psi_0)}= I_K^*\big[\overline{S_r(\psi_0)}\cdot\Phi_0\big]
   $ and $\nabla_{\psi}  f|_{(p_0,\psi_0)}=S_r^*\big[\overline{I_K(p_0)}\cdot\Phi_0\big].$
Turning now to the Hessian as a bilinear form, we apply \eqref{e2} to \eqref{tildef} (setting as before $a_0=I_K(p_0)$ and $b_0=S_r(\psi_0)$) to obtain 
\begin{multline}\label{equ:biH}
\HH^{ f}|_{ (p_0,\psi_0)}\Big((\Delta p^{(1)},\Delta \psi^{(1)}),(\Delta   p^{(2)},\Delta \psi^{(2)})\Big)=\hfill\\ 
\Re\Big\langle \nabla F|_{\LL(a_0,b_0)}, \mathrm{d}^2\LL|_{(a_0,b_0)}\Big(\big[I_K(\Delta p^{(1)}),S_r(\Delta \psi^{(1)})\big],\big[I_K(\Delta p^{(2)}),S_r(\Delta \psi^{(2)})\big]\Big)\Big\rangle +\\\HH^F|_{\LL(a_0,b_0)}\Big(\mathrm{d}\LL|_{(a_0,b_0)}\big[I_K(\Delta p^{(1)}),S_r(\Delta \psi^{(1)})\big],\mathrm{d}\LL|_{(a_0,b_0)}\big[I_K(\Delta p^{(2)}),S_r(\Delta \psi^{(2)})\big]\Big)
\end{multline}
where $\HH^F$ is given in \eqref{HbiFG}.

To wrap up, we have collected formulas for all partial derivatives and the bilinear Hessian, without explicitly differentiating a single time. The resulting expressions, when relying on the chain rule as above, are straightforward to implement and fast to execute, with all operators acting directly in matrix or tensor space. Armed with this, we can now approximate the graph of $f$ near a given point $x_0=(p_0,\psi_0)$, via $$f(x)\approx f(x_0)+\langle \nabla f|_{x_0},x-x_0\rangle +\frac{1}{2}\HH^f|_{ x_0}(x-x_0,x-x_0)$$ where $\langle \nabla f|_{x_0},x-x_0\rangle$ can be broken up into 
$$\langle \nabla f|_{x_0},x-x_0\rangle=\Re\langle \nabla_{p} f|_{x_0},p-p_0\rangle +\Re\langle \nabla_{\psi} f|_{x_0},\psi-\psi_0\rangle.$$

\subsection{ The Hessian operator for $f$}\label{sec:HesOp}

In order to implement second order solvers such as Newton's method, the bilinear Hessian is not enough  but we also need the Hessian operator. 
In this section we show how to retrieve an operator $H^f|_{(p_0,\psi_0)}:\C^{M\times M}\oplus \C^{N\times N}\rightarrow \C^{M\times M}\oplus \C^{N\times N}$ such that $$\HH^f|_{ (p_0,\psi_0)}\Big((\Delta p^{(1)},\Delta \psi^{(1)}),(\Delta p^{(2)},\Delta \psi^{(2)})\Big)=\Re\left \langle H^f|_{(p_0,\psi_0)} (\Delta p^{(1)},\Delta \psi^{(1)}),(\Delta p^{(2)},\Delta \psi^{(2)})\right\rangle. $$ 
Since the bilinear Hessian is symmetric, the operator is also symmetric. Note that, if we were to arrange all the elements of $\C^{M\times M}\oplus \C^{N\times N}$ in a column vector and separate the real and imaginary parts, thereby identifying $\C^{M\times M}\oplus \C^{N\times N}$ with $\R^{2M^2+2N^2}$, then $H^f|_{(p_0,\psi_0)}$ would really be the traditional Hessian matrix $\bold{H}^f|_{(p_0,\psi_0)}$. Also note that it naturally would split up into 4 submatrices \begin{equation}\label{gt6}
    \left(
                \begin{array}{cc}
                  \bold{H}^f_{11} & \bold{H}^f_{12} \\
                  \bold{H}^f_{21} & \bold{H}^f_{22} \\
                \end{array}
              \right)
\end{equation}
where $\bold{H}_{11}^f$ sends $p$-variables into new $p$-variables, $\bold{H}_{12}^f$ sends $\psi$-variables into $p-$variables and so on. By symmetry we would also have $\bold{H}_{21}^f={\bold{H}_{12}^f}^t$, where $t$ denotes matrix transpose. In a similar manner, we will find 4 operators $H^f_{11},H^f_{12}$ etc.~with corresponding roles, but acting on $\C^{M\times M}$ and $ \C^{N\times N}$ directly. These will be given by efficiently implementable rules and we shall have that $H^f_{21}$ is the adjoint of $H^f_{12}$ (with respect to the real scalar product on the respective spaces).

\begin{theorem}\label{t2}
The Hessian operator $H^f|_{(p_0,\psi_0)}(\Delta p,\Delta \psi)$ can be computed by the following steps: \begin{itemize}
\item Compute $\Phi_0={\DD_K}^*(\nabla F|_{\LL(p_0,\psi_0)})$ (recall \eqref{Phi}).
\item Compute $\Xi_0(\Delta p,\Delta \psi)={\DD_K}^*\Big[H^{F}|_{\LL(p_0,\psi_0)}\Big(\DD_K\big[\JJ(\Delta p)\cdot S_{{r}}(\psi_0)+ \JJ(p_0)\cdot S_{{r}}{(\Delta \psi)}\big]\Big)\Big]$
\item Compute $A_1=\JJ^*\Big(\overline{S_{{r}}({\Delta \psi})}\cdot\Phi_0+\overline{ S_{{r}}(\psi_0)}\cdot \Xi_0(\Delta p,\Delta \psi)\Big)  $
\item Compute $A_2=S_{{r}}^*\Big(\overline{\JJ(\Delta p)\cdot} \Phi_0+\overline{\JJ(p_0)}\cdot\Xi_0(\Delta p,\Delta \psi)\Big)  $
\item Set $H^f|_{(p_0,\psi_0)}=(A_1,A_2)$.
\end{itemize}
\end{theorem}

\begin{proof}
By combining formulas of Section \ref{sec:Objective} and \ref{man} we see that the bilinear Hessian is given by
\begin{equation*}\begin{aligned}
&\HH^{f}|_{ (p_0,\psi_0)}\Big((\Delta p^{(1)},\Delta \psi^{(1)}),(\Delta p^{(2)},\Delta \psi^{(2)})\Big)=\\&\Re \left\langle \Phi_0, \JJ(\Delta p^{(1)})\cdot S_{{r}}({\Delta \psi}^{(2)})+\JJ(\Delta p^{(2)})\cdot S_{{r}}({\Delta \psi}^{(1)})\right\rangle +\\& \Re \left\langle H^{F}|_{\LL(p_0,\psi_0)}\Big(\DD_K\big[\JJ(\Delta p^{(1)})\cdot S_{{r}}(\psi_0)+ \JJ(p_0)\cdot S_{{r}}{(\Delta \psi^{(1)})}\big]\Big),\right.\\&\quad \left.\DD_K\big[\JJ(\Delta p^{(2)})\cdot S_{{r}}(\psi_0)+ \JJ(p_0)\cdot S_{{r}}{(\Delta \psi^{(2)})}\big]\right\rangle=\\&
\Re \left\langle \Phi_0, \JJ(\Delta p^{(1)})\cdot S_{{r}}({\Delta \psi}^{(2)})+\JJ(\Delta p^{(2)})\cdot S_{{r}}({\Delta \psi}^{(1)})\right\rangle +\\& \Re\left\langle \Xi_0(\Delta p^{(1)},\Delta \psi^{(1)}),\JJ(\Delta p^{(2)})\cdot S_{{r}}(\psi_0)+ \JJ(p_0)\cdot S_{{r}}{(\Delta \psi^{(2)})}\right\rangle=\\&
\Re \left\langle \overline{S_{{r}}({\Delta \psi}^{(1)})}\cdot\Phi_0+\overline{ S_{{r}}(\psi_0)}\cdot \Xi_0(\Delta p^{(1)},\Delta \psi^{(1)}), \JJ(\Delta p^{(2)})\right\rangle+\\& \Re\left\langle \overline{\JJ(\Delta p^{(1)})\cdot} \Phi_0+\overline{\JJ(p_0)}\cdot\Xi_0(\Delta p^{(1)},\Delta \psi^{(1)}),  S_{{r}}({\Delta \psi}^{(2)})\right\rangle
\end{aligned}
\end{equation*} which combined with \eqref{Hessianop} easily yields the desired statement.
\end{proof}

We remark that if we consider the probe as fixed and only compute the Hessian operator related to $\psi$ (i.e. we insert $(0,\Delta\psi)$ above and only compute $A_2$, which is the operator corresponding to the matrix $\bold{H}^f_{22}$ had we vectorized), then the formula in Theorem \ref{t2} reduces to 
\begin{equation}\label{yeh}
A_2 = S_{{r}}^*\Big(\overline{\JJ(p_0)}\cdot{\DD_K}^*\Big[H^{F}|_{\LL(p_0,\psi_0)}\Big(\DD_K\big[ \JJ(p_0)\cdot S_{{r}}{(\Delta \psi)}\big]\Big)\Big]\Big).
\end{equation}
Analogous formulas have been computed in matrix form in \cite{yeh2015experimental}, see Appendix A. In particular, upon inserting formula \eqref{HFG} for $H^F$ in the above expression we retrieve the analogue of (42) from \cite{yeh2015experimental}.
Clearly, relying on \eqref{yeh} is much cleaner implementation-wise and many orders of magnitude faster to evaluate. Indeed, the numerical experiments in \cite{yeh2015experimental} are performed on $32\times32$ matrices, whereas our methods easily handle $2048\times2048$ images.

\section{Numerical experiments and results}
\label{sec:results}

We evaluate our approach using synthetic and experimental data reconstructions and compare the results to existing methods. Our goal is to introduce and share the new framework for utilizing second order information, not perform a detailed comparison to all analogs. There are many other reconstruction methods in ptychography, and their behavior may vary for different datasets. We defer a detailed comparison to future work.
Here we compare the convergence behavior and reconstruction results of the proposed BH-GD, BH-CG, and BH-QN methods with state-of-the-art Least Squares Maximum Likelihood (LSQML)\cite{odstrvcil2018iterative} method and Automatic Differentiation (AD) with Adam optimizer~\cite{kandel2019using}, referred to as “Adam”. LSQML is a generalized ptychography approach that optimizes the reconstruction simultaneously for the object, probe, and positions. It employs an optimized strategy to calculate the step length in the gradient descent direction for each variable, avoiding computationally expensive line search procedures. The authors have shown that LSQML outperforms ePIE~\cite{maiden2009improved} and Difference Map (DM)~\cite{thibault2008high}. 
Reconstruction using different methods was performed within the same framework, employing identical forward and adjoint Fresnel transform operators, shift operators, and other computational components. All methods were implemented in Python, with GPU acceleration. A single NVIDIA Tesla A100 was used for conducting performance tests.

\subsection{Synthetic data: object and probe retrieval}

As a synthetic data example, we generated an object based on the commonly used Siemens star sample, which is often employed at synchrotron beamlines to test spatial resolution. The phase component of the generated object is shown in the top-left corner of Figure~\ref{fig:fig_syn_rec}. The amplitude component is 30 times smaller. The object was created as a set of triangles with gaps at varying distances from the origin. To increase the complexity of the object, we randomly inserted sharp rectangles of different sizes within the star's segments. Additionally, we introduced low-frequency components in the background, as recovering these components may be more challenging in near-field ptychography. 

\begin{figure}
    \includegraphics[width=1\linewidth]{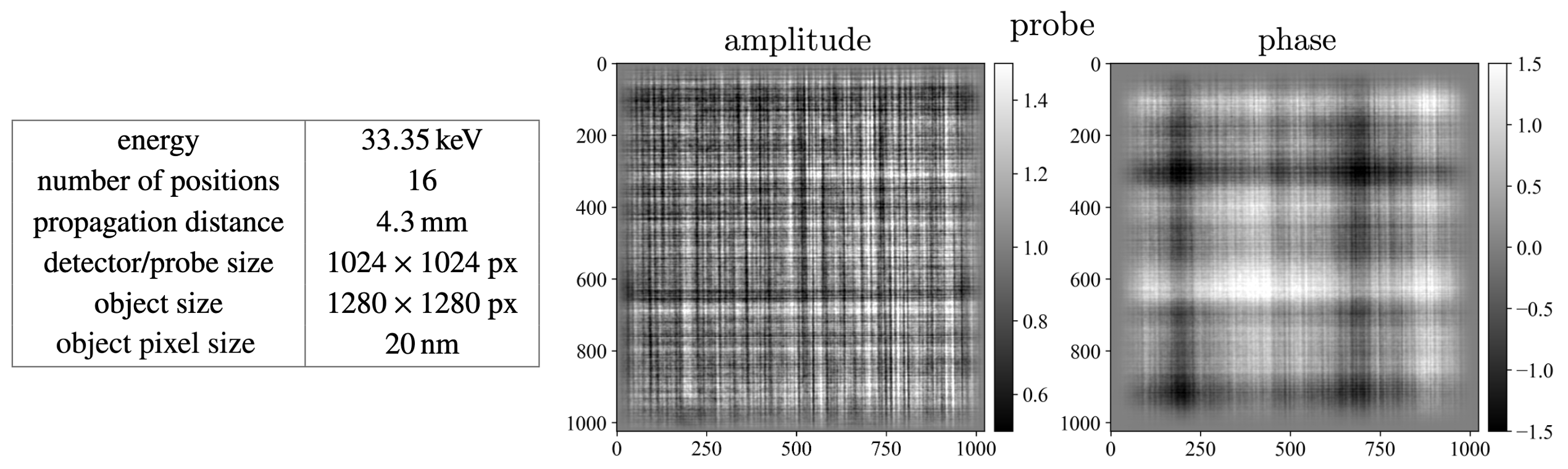}  \\
    \caption{Settings and the probe for near-field ptychography simulations.}
    \label{fig:settings_syn}
\end{figure}

The data modeling settings are shown in the left panel of Figure~\ref{fig:settings_syn}, replicating real-world parameters commonly encountered when working with experimental data. For the probe, we used a previous probe reconstruction from experimental data acquired at beamline ID16A of the ESRF. The amplitude and phase of the probe are displayed in the right panel of Figure~\ref{fig:settings_syn}. We note that the probe exhibits prominent vertical and horizontal components introduced by imperfections of the
multilayer coated KB mirrors. A strong probe modulation is required for the near-field ptychography method to work. The data were modeled for 16 object positions chosen based on the strategy described in~\cite{stockmar2013near}. Additionally, Gaussian noise with an SNR of 60 dB was added to the simulated data.

The simultaneous recovery of the object and probe was achieved using different methods. Position correction procedures were omitted as they vary across packages and are often separate steps from object and probe optimization. An initial estimate for reconstructing the sample transmittance function $\psi$ was obtained using the Transport of Intensity Equation (TIE) method~\cite{gureyev1996phase,paganin2006coherent}, commonly known as the Paganin method in synchrotron beamline applications. 
This method was applied to the data after normalization by the reference image (i.e., data acquired without a sample). It is important to note that the Paganin method generally struggles to recover high-frequency components. As a result, the presence of horizontal and vertical line artifacts in the data due to the probe shape does not significantly affect the reconstruction quality. For the initial estimate of the probe function $p$, we used the square root of the reference image  propagated back to the sample plane. The preconditioning factors for the variables $\psi$ and $p$, as described in Appendix II, were experimentally determined to be 1 and 2, respectively.

Figure~\ref{fig:fig_syn_plots} shows convergence plots for the objective functional value over 1000 iterations, using different optimization methods. In the left plot, the x-axis represents the iteration number, while in the right plot it represents computation time (in seconds) for a fair performance comparison between methods. Both plots feature insets that allow for a detailed analysis of the algorithm's behavior at the beginning, during iterations 5 to 70.

Within the illuminated region, all reconstruction methods converge to an approximation of the ground truth object, with numerical precision achieved only after as many as 20000 iterations for the slower methods. For demonstration purposes, in Figure~\ref{fig:fig_syn_rec} we compare the reconstruction results each algorithm produces after 6 seconds of execution and compared them to the ground truth. During this time, the algorithms perform different numbers of iterations; the corresponding states of the objective functional are marked by the "visualization" line in the right panel of Figure~\ref{fig:fig_syn_plots}. For instance, BH-QN completes about 20 iterations, while BH-GD completes around 120 iterations.  
For reference, the reconstructed probe after 6 seconds using the BH-CG method is shown in Figure~\ref{fig:fig_rec_probes} a).

\begin{figure}[t]
    \centering
    \includegraphics[width=1\linewidth]{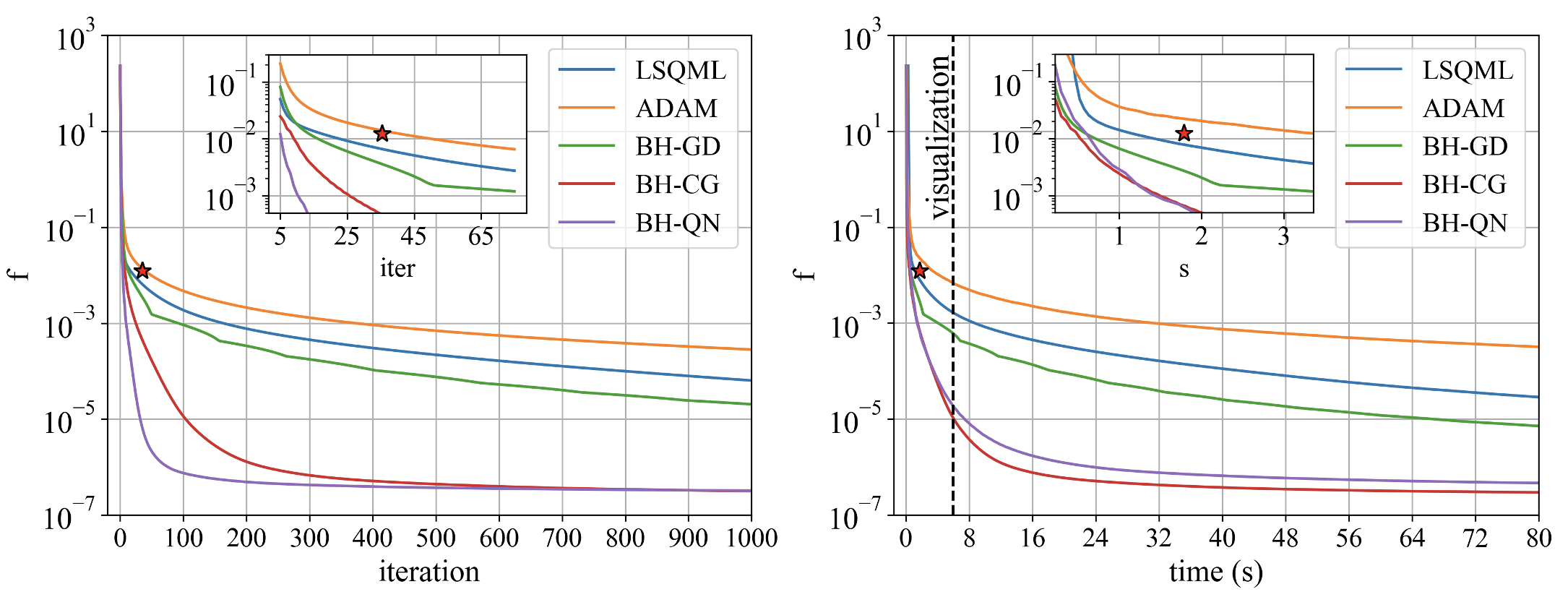}
    \caption{Objective functional value vs. iteration number (left) and vs. computation time (right) when reconstructing the object and probe from synthetic data. }
    \label{fig:fig_syn_plots}
\end{figure}

\begin{figure}[t]
    \centering
    \includegraphics[width=1\linewidth]{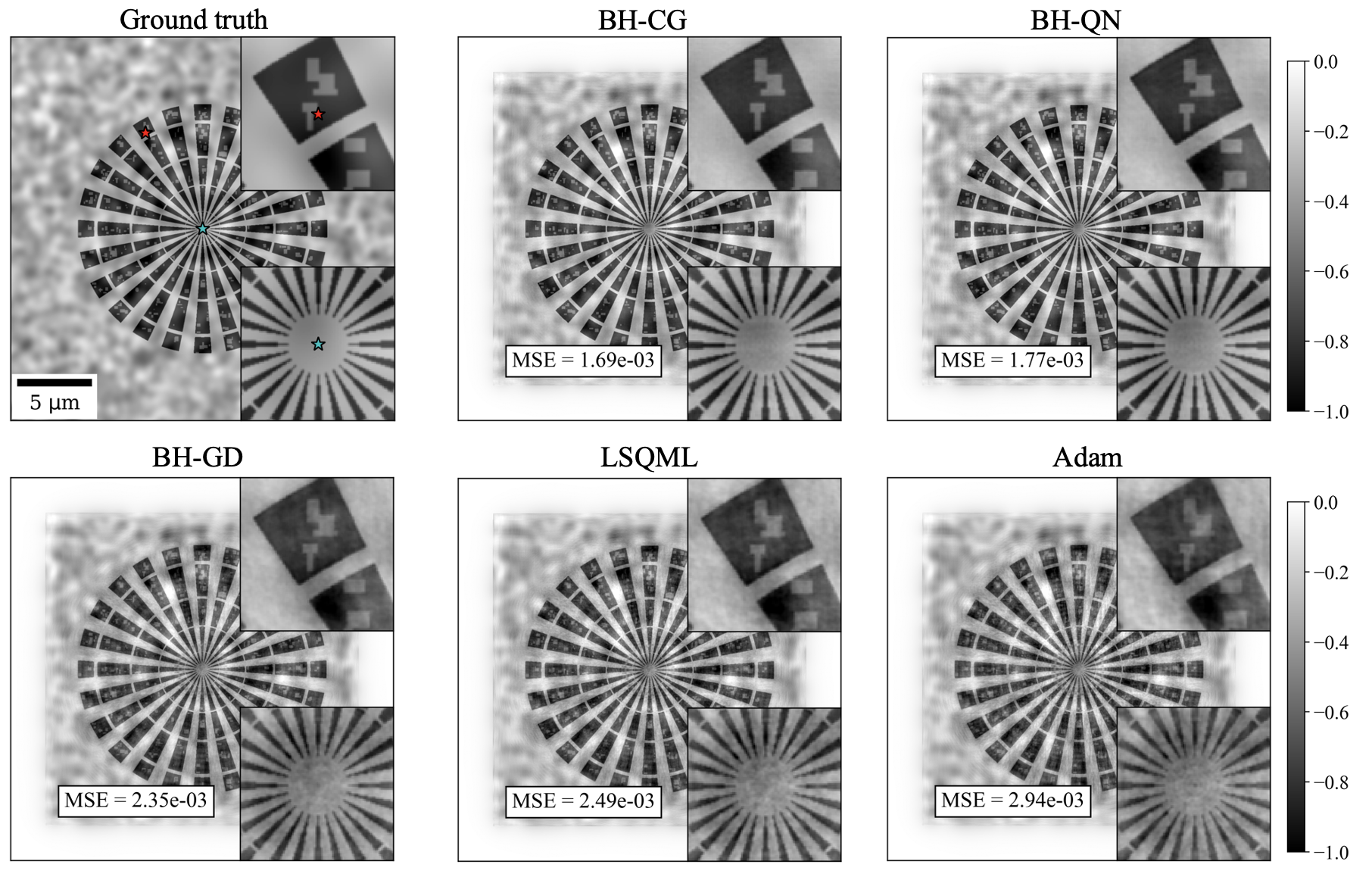}
    \caption{Reconstruction results (object phase) for a synthetic siemens star after 6 s execution of different methods. Corresponding states of the objective functional are marked with line 'visualization' in Figure~\ref{fig:fig_syn_plots}.}
    \label{fig:fig_syn_rec}
\end{figure}

Several key observations can be made from these simulations. First, the second-order methods, BH-QN and BH-CG, outperform the others by a factor of at least 8-10. The reconstruction quality after 6 seconds of iterative scheme execution is significantly better for the second-order methods, as evidenced by the inset regions in Figure~\ref{fig:fig_syn_rec}. This is further confirmed by the Mean Squared Error (MSE), displayed within the figure, which compares the illuminated region to the ground truth.

Second, although both LSQML and BH-GD use gradient descent steps, BH-GD exhibits faster convergence. This is likely because BH-GD employs a direct formula for calculating a joint step size for all variables (cf. formula~\eqref{eqalpha}), while LSQML calculates the step size approximately and independently for each variable, see~\cite{odstrvcil2018iterative} for details. 

Third, although the Adam method uses a constant gradient step length manually adjusted for optimal convergence speed, it ultimately performs slower than methods that adjust the step size for each iteration. This is evident when comparing the Adam, LSQML, and BH-GD plots in the left and right panels. This highlights the disadvantage of automatic differention methods in achieving faster convergence, particularly in more complex problems. 

Finally, it is worth comparing the two leading methods, BH-CG and BH-QN. While BH-QN requires fewer iterations to converge, BH-CG shows slightly faster performance when considering the functional vs. time plot. This is because BH-QN involves solving problem \eqref{QNeq1} iteratively. In these tests we solve the problem approximately by gradually increasing the number of inner iterations. Specifically, we set the number of inner iterations to $\text{max}(5,\lfloor \frac{k}{2} \rfloor)$, where $k$ represents the outer iteration number in the main optimization problem. Additionally, to avoid instabilities of BH-QN, we perform the first three outer iterations using BH-CG. We believe this strategy could be further optimized, for example by monitoring the residual at each inner iteration and stopping when the residual falls below a certain threshold. 
Nevertheless, since BH-QN is also more sensitive to initialization, we conclude that BH-CG is overall to be preferred in practice.
\begin{figure}[t]
    \centering
    \includegraphics[width=1\linewidth]{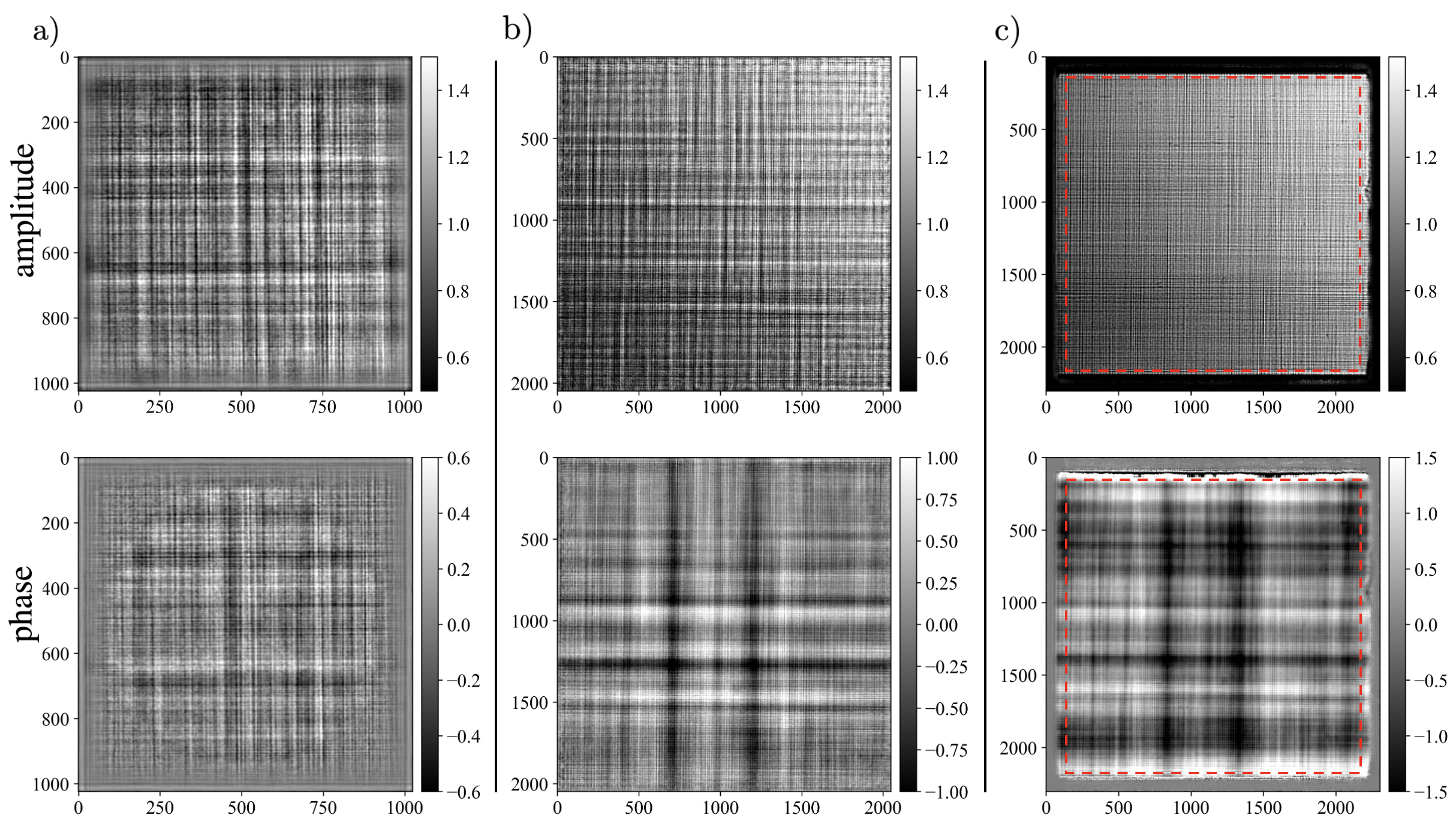}
    \caption{Probes recovered during reconstruction with BH-CG of a) synthetic data from Figure~\ref{fig:settings_syn}, b) experimental Siemens star data from Figure~\ref{fig:settings_real}, and c) experimental coded aperture data from Figure~\ref{fig:settings_real_ca}. The red dashed frame in c) outlines the detector size.}
    \label{fig:fig_rec_probes}
\end{figure}

\subsection{Experimental data: object and probe retrieval}
As the first experimental dataset, we consider measurements of a 200 nm thick gold Siemens star at beamline ID16A of the ESRF. This object is routinely used at the beamline for optics calibration and resolution tests. Acquisition parameters, along with an example of the acquired data, are shown in Figure~\ref{fig:settings_real}. The data was measured for 16 object positions chosen based on the strategy described in~\cite{stockmar2013near}. 
Unlike the synthetic tests, where the propagation distance was specified initially, here we begin with the distances defined in cone beam geometry and convert them to parallel beam geometry. This conversion ensures proper rescaling of coordinates based on the Fresnel scaling theorem~\cite{paganin2006coherent}. Further details can be found in the Appendix of ~\cite{nikitin2024x}.

\begin{figure}[t]  
    \includegraphics[width=1\linewidth]{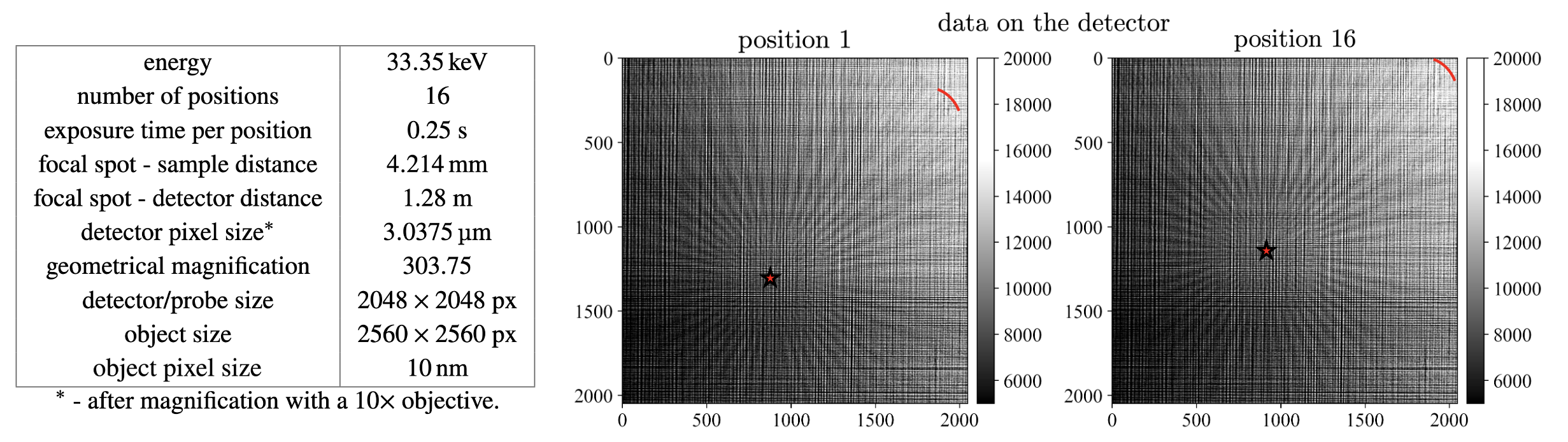}  \\
    \caption{Data acquisition parameters and examples of measured data  for near-field ptychography of the 200 nm thick gold Siemens star at beamline ID16A of ESRF.}
    \label{fig:settings_real}
\end{figure}

In this experiment, the linear stages moving the sample laterally were highly precise. Preliminary reconstruction using the BH-CG method with simultaneous position correction showed that the positional errors were less than 0.15 pixels. Consequently, we performed reconstructions with various methods without position correction, similar to the approach used for the synthetic data in the previous section. Prior to reconstruction, we also applied "zinger removal", a common pre-processing procedure including the conditional median filter (applied when the value of the pixel is significantly different from the median) to eliminate isolated bright pixels which can result from malfunctioning detector regions or parasitic scattering.

Similar to the synthetic data test, for the initial estimate of the sample transmittance function $\psi$, we used the reconstruction obtained from the Paganin method, applied to the data divided by the reference image. For the initial estimate of the probe function $p$, we used the square root of the reference image, propagated back to the sample plane.  The scaling factors for the variables $\psi$ and $p$ were chosen as 1 and 2, respectively. 

Convergence plots in Figure~\ref{fig:fig_real_plots} again demonstrate that second-order methods significantly outperform the other ones: BH-CG is more than 10 times faster than BH-GD. LSQML and BH-GD demonstrate similar convergence rate during first 200 iterations, however, iterations with the BH-GD is faster yielding better overall performance. Again, Adam is found to yield very slow convergence.

\begin{figure}[h!]
    \centering
    \includegraphics[width=1\linewidth]{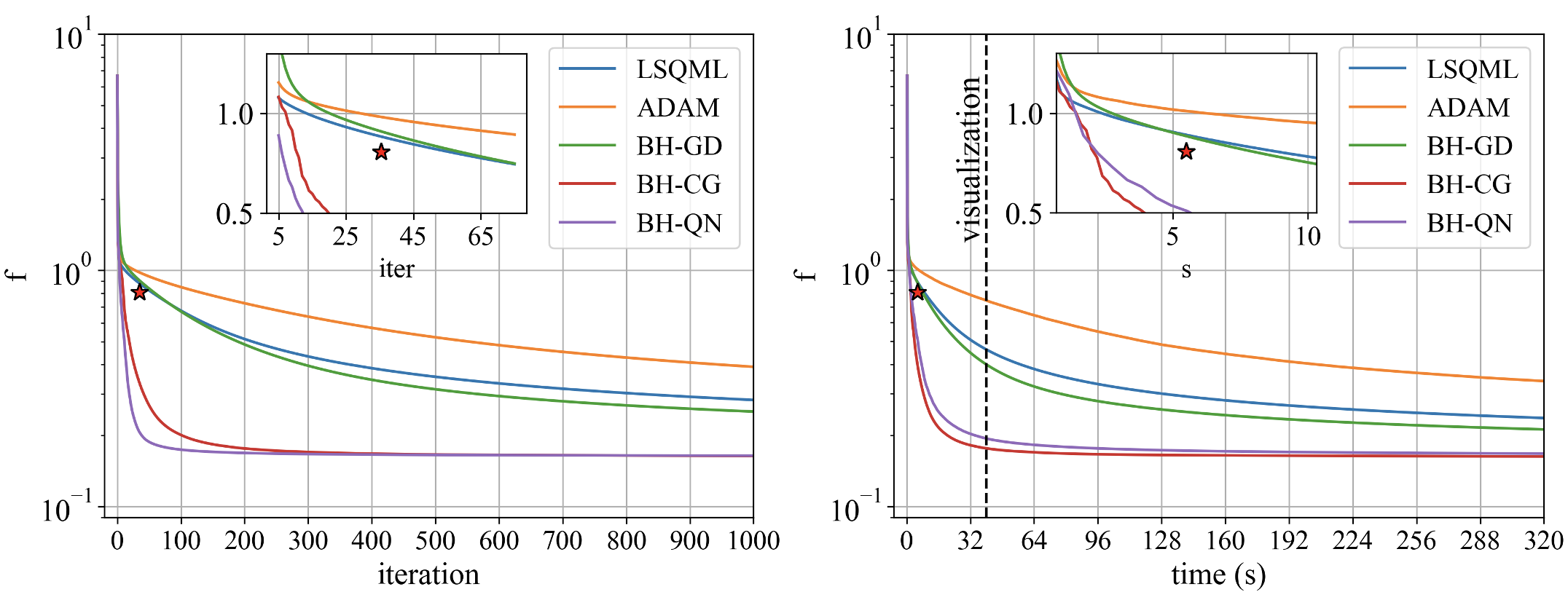}
    \caption{Objective functional value vs. iteration number (left) and vs. computation time (right) when reconstructing the object and probe from experimental Siemens star data acquired at ID16A of ESRF. }
    \label{fig:fig_real_plots}
\end{figure}

There is a difference in performance of BH-CG and BH-QN for this experimental dataset. While BH-QN converges in fewer iterations, its overall performance appears slower—roughly twice as slow as BH-CG. As seen in the right panel of Figure~\ref{fig:fig_real_plots}, BH-CG reaches the bottom of the plot in 32 seconds, whereas BH-QN reaches it in 64 seconds.

\begin{figure}[h!]
    \centering
    \includegraphics[width=1\linewidth]{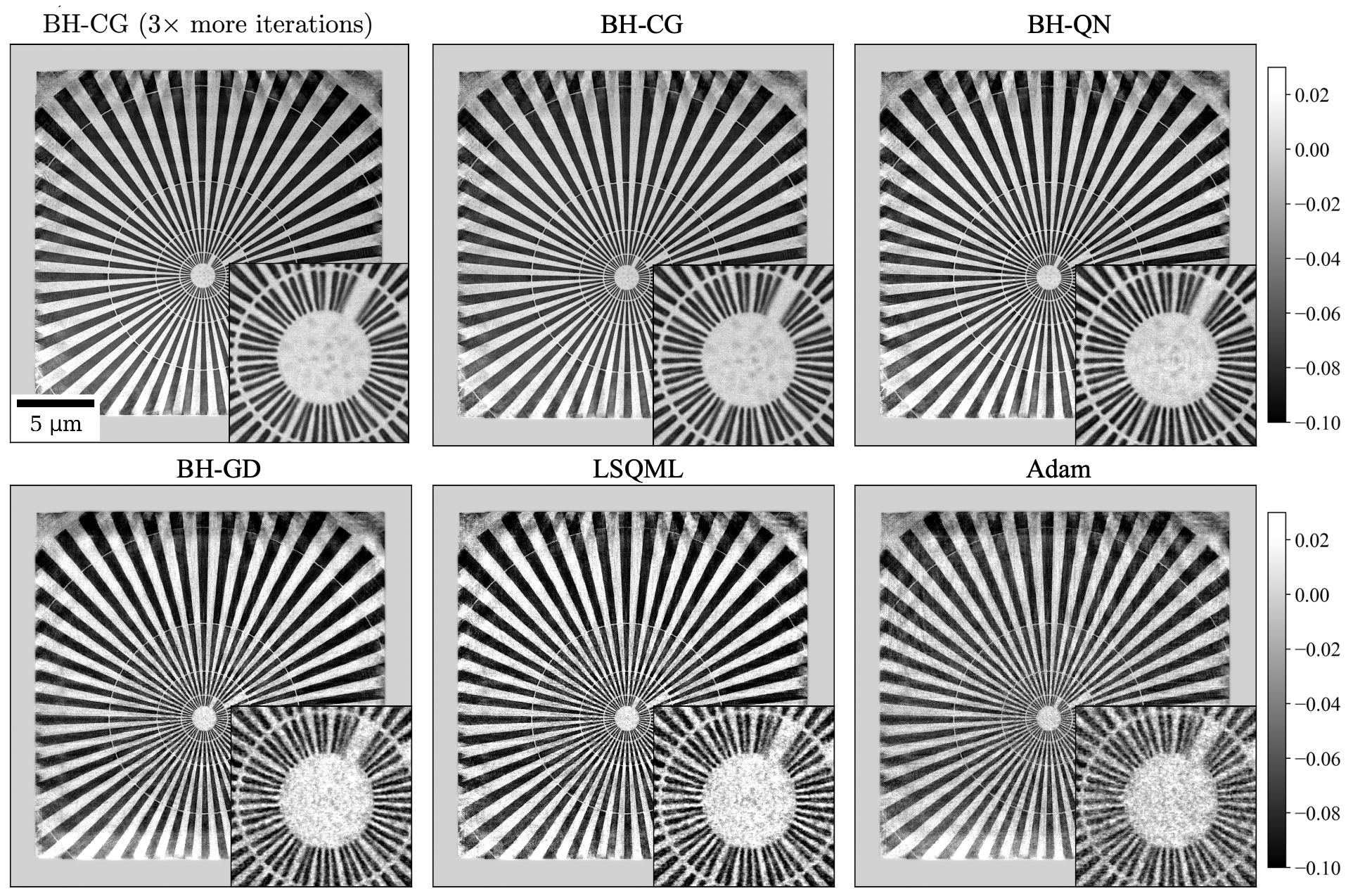}
    \caption{Reconstruction results (phase of the object) for experimental gold Siemens star data from ID16A of ESRF after 40 s execution of different methods and 120 s ($3\times $ more iteration) of BH-CG for reference. Corresponding states of the objective functional are marked with line 'visualization' in Figure~\ref{fig:fig_real_plots}.}
    \label{fig:fig_real_rec}
\end{figure}

Figure~\ref{fig:fig_real_rec} shows the reconstruction results after executing each method for 40 seconds, with an additional image showing 120 seconds ($3\times$ more iterations) of execution for BH-CG as a reference. The reconstructed probe is shown in Figure~\ref{fig:fig_rec_probes} b). Corresponding states of the objective functional are marked with line 'visualization' in Figure~\ref{fig:fig_real_plots}.

Both BH-CG and BH-QN produce high-resolution results, with the smallest features of the Siemens star clearly separable and fabrication defects in the center visible. BH-QN exhibits a noticeable difference in the middle region compared to BH-CG. This difference disappears after 60 s execution of BH-QN. In contrast, the reconstructions from BH-GD, LSQML, and Adam are noisy and far from converged, with small features not visible.

\subsection{Experimental data: object, probe, and position correction}

In this section, we demonstrate the simultaneous reconstruction of the object, probe, and position correction, as described in Appendix I, using experimental data from a coded aperture sample collected at ID16A of ESRF. The data was acquired as part of a project developing a single-distance holotomography method using coded apertures~\cite{nikitin2025single}. The coded aperture used is a binary gold mask with a 2 $\mu$m bin size and 2 $\mu$m thickness. Compared to the 200 nm gold Siemens star sample from the previous section, the coded aperture introduces larger phase shifts in the wavefront, making it more suitable for characterizing the illumination structure, i.e., for reconstructing the probe.

\begin{figure}[t]  
\centering
    \includegraphics[width=0.8\linewidth]{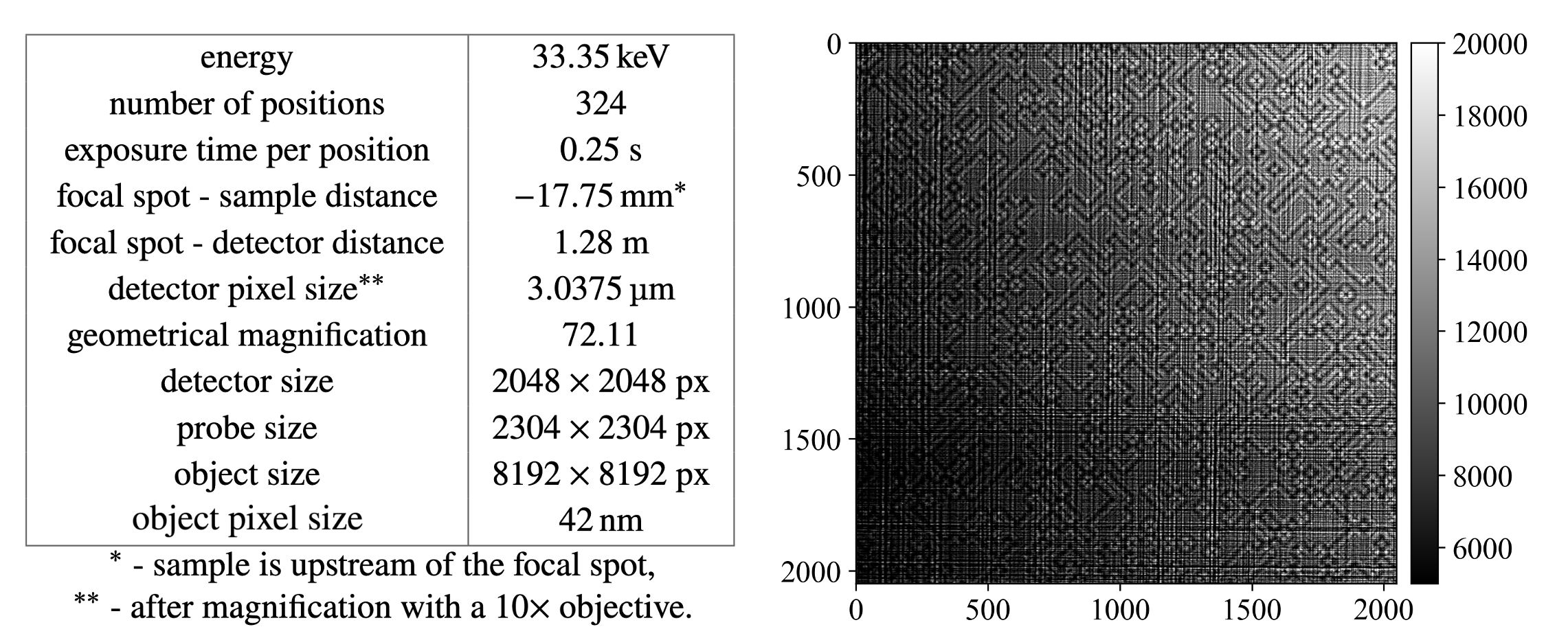}  \\
    \caption{Data acquisition parameters and an example of measured data  for near-field ptychography of the coded aperture at beamline ID16A of ESRF.}
    \label{fig:settings_real_ca}
\end{figure}

The data acquisition settings and an example of the acquired data are shown in Figure~\ref{fig:settings_real_ca}. Notably, unlike the Siemens star experiment, the coded aperture was positioned upstream of the focal spot. This configuration enhances structured illumination more effectively than placing it downstream while also preserving space for positioning the actual sample.
 While studying structured illumination is beyond the scope of this work, it is important to note that for the approach with the coded aperture to work, accurate reconstruction of both the probe and the coded aperture is essential. This was achieved using the near-field ptychography method that we study in this work. In this experiment, the coded aperture movement was not controlled by precise motors, so position refinement is necessary. 

Before reconstruction we obtained a coarse estimate for the position errors using cross-correlation of adjacent diffraction patterns. As with the Siemens star dataset, we applied zinger correction and computed the initial guess for the sample using the Paganin method, while the initial probe estimate was obtained by backpropagating the square root of the reference image.
The scaling factors for the variables $\psi,p$ and $r$ were chosen as 1,2, and 0.1, respectively.

In contrast to the Siemens star experiment, the probe reconstruction was performed on a grid larger than the detector size. The probe extension region on each side is approximately three times the size of the first Fresnel zone.
Specifically the probe size is $2304 \times 2304$, while the detector is $2048\times 2048$. The object, with a size $0.34\times 0.34$ mm, was reconstructed on a significantly larger grid, $8192 \times 8192$, using $18\times 18$ uniformly distributed positions. 

For demonstration, we perform the reconstruction using only the BH-CG method and analyze the convergence behavior with and without position correction. We do not compare this method to others, as different implementations of position correction exist across various packages, often functioning as independent supplementary steps rather than being integrated into the object and probe optimization process.

It is important to note that the entire dataset and auxiliary variables in the BH-CG method do not fit into GPU memory. As a result, we implemented data chunking and optimized data transfers between the CPU and GPU for processing. For reference, the reconstruction time for 150 iterations using the BH-CG method was approximately 1 hour.

\begin{figure}[t]
    \centering
    \includegraphics[width=1\linewidth]{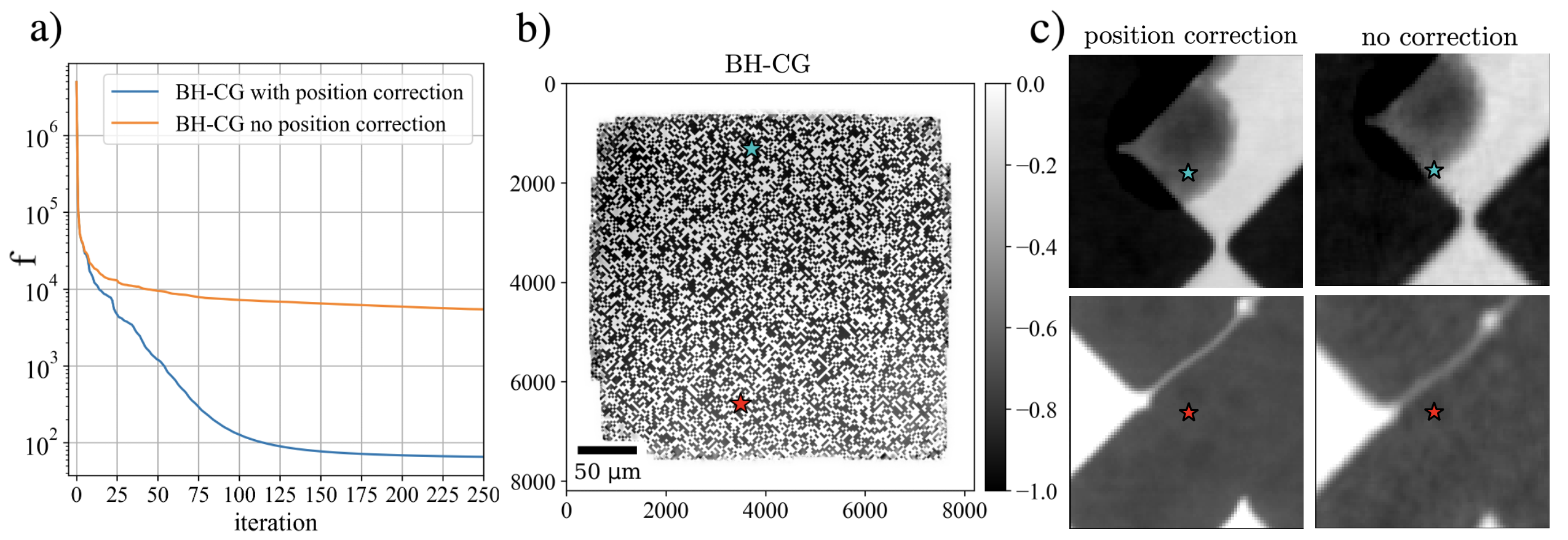}
    \caption{a) Objective functional value vs. iteration number when reconstructing the coded aperture object and probe with and without position correction from data acquired at ID16A of ESRF. b) Phase of the entire reconstructed object. c) Magnified regions indicated with stars in b) comparing reconstruction with (left) and without position correction. }
    \label{fig:fig_rec_ca}
\end{figure}

The reconstruction results for the coded aperture sample are shown in Figure~\ref{fig:fig_rec_ca}. In this figure we compare reconstructions using the proposed scheme with and without positions correction inside the BH-CG iterative scheme. The object positions for the later case were determined using cross-correlation of adjacent diffraction patterns before reconstruction. The convergence plot in the left part of the figure indicates that when the positions are also being corrected for during reconstruction, the BH-CG method minimizes the objective functional to $10^{-5}$ times its initial value within 150 iterations. In contrast, when positions are not corrected, the objective functional remains at a 100 times larger value. Reconstruction results, shown in the right part of the figure, display the full coded aperture image and zoomed-in regions containing various features. It is clear that position correction significantly enhances image sharpness, thereby improving spatial resolution. The found floating-point position errors in the horizontal and vertical directions are in the range $(-8,8)$, demonstrating the effectiveness of the proposed method compared to the approach where the positions were found as a pre-processing step. 

Finally, it is interesting to analyze the reconstructed probe in this experiment, shown in Figure~\ref{fig:fig_rec_probes} c). The probe is accurately reconstructed even outside the field of view (marked with red color in the figure). The slits that limit the beam illumination are clearly visible as black regions on either side of the probe. These slits do not appear in the measured data (right part of Figure~\ref{fig:settings_real_ca}), yet they are successfully recovered due to interference fringes appearing inside the detector field of view, but originating from features outside the field of view.

\section{Conclusion}\label{sec:conclusion}


In this work, we introduced an innovative ptychography reconstruction method that leverages the bilinear Hessian and Hessian operators to significantly accelerate the convergence and improve the accuracy of iterative schemes for refining the object, probe, and position parameters. We provided the necessary mathematical formulations to implement gradient descent, conjugate gradient, and quasi-Newton schemes. 
Deriving gradients and operators for ptychographic reconstruction can result in complex expressions, making implementation challenging. While some researchers advocate for Automatic Differentiation (AD) as a more convenient alternative \cite{jurling2014applications,kandel2019using}, it requires substantial computational resources and is often significantly slower \cite{srajer2018benchmark}. In \cite{Carlsson:25}, computing the Hessian operator with AD was found to be up to an order of magnitude slower, even for simple cases, a result confirmed by our numerical findings. Moreover, AD struggles with essential modifications like preconditioning \cite{kandel2021efficient}, failing to eliminate the need for manual computations. In contrast, our framework provides a systematic approach to deriving gradients and operators, producing transparent, easily implemented expressions through derivative-free computations and the chain rule for the bilinear Hessian. Given these advantages, we argue that manual computation using our framework is a more efficient and practical solution than AD.

As a demonstration of the effectiveness of our proposed method, we applied it to both synthetic and experimental near-field ptychography data. The proposed second-order methods outperform traditional techniques by one order of magnitude, showcasing their remarkable efficiency in handling large-scale ptychography datasets. Since the performance of the proposed conjugate gradient (BH-CG) and Quasi-Newton (BH-QN) method are very similar, we advocate for the use of BH-CG since Quasi-Newton is known to be more sensitive to initialization and more complex, requiring additional parameter tuning, whereas our proposed BH-CG is parameter free and simple to implement and to use.
The performance improvements offered by our method represent a significant advancement in the field of ptychography, where high-quality reconstructions are often computationally expensive and time-consuming. Conventional phase retrieval reconstruction in 3D ptychography, which involves capturing data at different object rotation angles, can take up to a week using existing methods. However, our proposed second-order techniques achieve the same results in less than a day, offering a substantial reduction in computational time. This improvement is particularly impactful for synchrotron beamlines and in-situ ptychography experiments, where immediate feedback is crucial for adjusting environmental conditions and optimizing sample preparation. 

Looking ahead, a key goal is adapting our method for processing experimental far-field ptychography data. While our approach is valid for this setting and also adaptable to Poisson noise, far-field ptychography presents challenges that are distinct from near-field ptychography. Further advancements will include multimode probe reconstruction for improved robustness in the presence of complex sample geometries, and orthogonal probe relaxation~\cite{odstrcil2016ptychographic} to reduce artifacts in the reconstructed images. Additionally, we plan to explore batch data processing to further accelerate convergence~\cite{odstrvcil2018iterative,tripathi2024stochastic}. 
While integrating multimode probes and orthogonal relaxation should be straightforward, batch processing with conjugate gradient methods remains largely unexplored and presents an exciting challenge and valuable addition to our toolkit.

We are also exploring the application of our method to 3D ptychography, where 2D ptychography and tomography are solved jointly. As shown in previous studies \cite{gursoy2017direct,Aslan:19, nikitin2019photon}, joint reconstruction improves image quality while reducing the required scanning positions and rotation angles. This not only lowers measurement overhead but also minimizes radiation exposure, which is crucial for preserving sensitive samples in biological imaging and materials science.

Beyond methodological challenges, handling the large datasets of 3D ptychography requires efficient computational strategies. Optimized distributed computing across multiple GPUs and nodes is essential to keep reconstruction feasible for complex experiments with massive data volumes. Additionally, enhancing memory efficiency and parallelization will be crucial for scaling to even larger datasets.

\section*{Acknowledgments}

This research used resources of the Advanced Photon Source, a U.S. Department of Energy (DOE) Office of Science user facility and is based on work supported by Laboratory Directed Research and Development (LDRD) funding from Argonne National Laboratory, provided by the Director, Office of Science, of the U.S. DOE under Contract No. DE-AC02-06CH11357. We acknowledge the European Synchrotron Radiation Facility (ESRF) for provision of synchrotron radiation facilities  under proposal number MI-1506. Work in part also supported by the Swedish Research Council Grant 2022-04917.

\section*{Disclosures}
The authors declare no conflicts of interest.

\section*{Data availability}
Data underlying the results presented in this paper are not publicly available at this time but may be obtained from the authors upon reasonable request.

\section*{Code}
Demonstration code is available under \href{https://github.com/nikitinvv/BH-ptychography/tree/paper}{https://github.com/nikitinvv/BH-ptychography/tree/paper}.

\section*{Appendix I: The bilinear Hessian for ptychography with position correction}\label{sec:positions}

In this section we modify the objective functional $f$ defined in \eqref{fG} so that also the positions $r$ are considered as unknowns. By slight abuse of notation, we will simply add $r$ as an independent variable and write
\begin{equation}\label{eqposcor}
f(p,\psi,
{r})=F\big(\LL(\JJ(p), S_{{r}}(\psi)\big)
\end{equation}
just like in \eqref{tildef}. The shift operators $S_{{r}}$ are defined in the Fourier domain, i.e., for each slice $(S_{{r}}(\psi))_k$, $k=1,\ldots,K$,  $$(S_{{r}}(\psi))_k=\CC(\mathcal{F}^{-1}(e^{ - 2\pi i \langle {r}_k, {{\xi}}\rangle}\hat{\psi}({\xi})))$$  where 
\begin{itemize}
  \item $\mathcal{F}$ denotes the (discrete) Fourier transform and we also write $\hat{\psi}$ in place of $\mathcal{F}(\psi)$,
  \item $\langle {r}_k, {{\xi}}\rangle=r_{k,1}\xi_1+r_{k,2}\xi_2$ is the scalar product in $\R^2$,
   \item $\CC$ is an operator which crops the full image $\C^{N\times N}$ to a smaller central piece $\C^{M\times M}$.
\end{itemize}  
It will be convenient to think of the above operator also as a function of ${r}_k$ so therefore we introduce $\TT_\psi:\R^2\rightarrow \C^{M\times M}$ defined by
$$\TT_\psi({z})=\CC\Big(\mathcal{F}^{-1}\big[e^{ - 2\pi i \langle {z}, {\xi}\rangle}\hat{\psi}({\xi})\big]\Big).$$
We thus have $S_r(\psi)=\big[\TT_{\psi}(r_k)\big]_{k=1}^K$; though this introduces two notations for essentially the same operation, it will significantly simplify the upcoming computations.
We first expand $\TT_\psi({z})$ in a Taylor-series with respect to a perturbation  $\Delta {z}$ of some fixed ${z}_0$. Since $$e^{-2\pi i  \langle {z}_0+\Delta {z}, {\xi}\rangle}=e^{-2\pi i \langle {z}_0, {\xi}\rangle}\left(1- 2\pi i \langle \Delta {z},{\xi}\rangle+\frac{(- 2\pi i )^2}{2}\langle \Delta {z}, {\xi}\rangle^2+\mathcal{O}\left(\|\Delta {z}\|^3\right)\right)$$ we get that 
\begin{align*}
  \TT_{\psi}({z}_0+\Delta {z})=&\CC(\mathcal{F}^{-1}\big[e^{- 2\pi i \langle {z}_0, {\xi}\rangle}\Big(1- 2\pi i  \langle\Delta {z}, {\xi}\rangle-{2}\pi^2\langle  \Delta {z},{\xi}\rangle^2+\mathcal{O}(\|\Delta {z}\|^3)\Big)\hat{\psi}({\xi})\big]).
\end{align*}
Comparing this expression with \eqref{expdef}
we see that
\begin{equation}\label{l04}
\mathrm{d}\TT_\psi|_{{{z}_0}}(\Delta {z})=-2\pi i \CC\Big(\mathcal{F}^{-1}\big[\langle\Delta{z},{\xi}\rangle e^{-2\pi i  \langle {z}_0, {\xi}\rangle}\hat{\psi}({\xi})\big]\Big)
\end{equation}
and
\begin{align*}
  &  \mathrm{d}^2{\TT_\psi}|_{{{z}_0}}(\Delta{z},\Delta{z})=-4\pi^2\CC\Big(\mathcal{F}^{-1}\big[\langle\Delta{z},{\xi}\rangle^2 e^{-2\pi i  \langle { z}_0, {\xi}\rangle}\hat{\psi}({\xi})\big]\Big).
\end{align*}
Since later we will also need the bilinear version of this, we remark already that 
\begin{equation}\label{e34}
   \mathrm{d}^2{\TT_\psi}|_{{{z}_0}}(\Delta{z}^{(1)},\Delta{z}^{(2)})=-4\pi^2\CC\Big(\mathcal{F}^{-1}\big[\langle\Delta{z}^{(1)}, {\xi}\rangle\langle\Delta{z}^{(2)}, {\xi}\rangle e^{- 2\pi i  \langle{z}_0, {\xi}\rangle}\hat{\psi}(\xi)\big]\Big),
\end{equation}
is real bilinear and symmetric, so by uniqueness of such objects it must be the one sought.

Since the operator $\TT_\psi$ is linear in the $\psi$-variable, we get that the joint expansion for a perturbation $\TT_{\psi_0+\Delta \psi}({z}_0+\Delta {z})$ considered as a function on $\C^{N\times N}\times \R^2$ becomes 
\begin{align*}
  \TT_{\psi_0+\Delta \psi}({z}_0+\Delta {z})=&\TT_{\psi_0}({z}_0+\Delta {z})+\TT_{\Delta \psi}({z}_0+\Delta {z})=
 \TT_{\psi_0}({{z}_0})+\mathrm{d}\TT_{\psi_0}|_{{{z}_0}}( \Delta {z}) +\TT_{\Delta \psi}({{z}_0})+  \\& \mathrm{d}\TT_{\Delta \psi}|_{{{z}_0}}( \Delta {z})+\frac{1}{2}\left(\mathrm{d}^2{\TT_{\psi_0}}|_{{{z}_0}}(\Delta{z},\Delta{z})\right)+\mathcal{O}\left(\|(\Delta \psi,\Delta {z})\|^3\right).
\end{align*}
Here, first order terms are collected on the first row, and the second row contains the second order terms and the ordo. 
We now consider the function $\SS:\C^{N\times N} \times (\R^2)^K\rightarrow \C^{M\times M\times K}$ defined as $$\SS(\psi,r)= S_{{r}}({\psi})=\left(\TT_{\psi}({r}_k)\right)_{k=1}^K,$$ where ${r}=({r_1},\ldots,{r}_K)$ and  ${r}_k\in\R^2$. By the linear part of the penultimate equation we have that \begin{align}
 & \mathrm{d}\SS|_{(\psi_0,r_0)}(\Delta \psi, \Delta r)=\Big(\big(\TT_{\Delta \psi}({{r}_{0,k}})+  \mathrm{d}\TT_{\psi_0}|_{{{r}_{0,k}}}( \Delta {r}_k)\big) \Big)_{k=1}^K=\nonumber S_{{r}_{0}}{(\Delta \psi)}+  \Big( \mathrm{d}\TT_{\psi_0}|_{{{r}_{0,k}}}( \Delta {r}_k) \Big)_{k=1}^K\label{DLL4}
\end{align}
and the bottom row gives 
\begin{align*}
 & \mathrm{d}^2\SS|_{(\psi_0,r_0)}\big(( \Delta \psi, \Delta r),(\Delta \psi, \Delta r)\big)=\Big( 2 \mathrm{d}\TT_{\Delta \psi}|_{{{r}_{0,k}}}( \Delta {r}_k)+ \mathrm{d}^2{\TT_{\psi_0}}|_{{{r}_{0,k}}}(\Delta{r}_k,\Delta{r}_k)\Big)_{k=1}^K
\end{align*}
where $\mathrm{d}^2{\TT_\psi}|_{{{r}_{0,k}}}$ is given in \eqref{e34}. The bilinear version is now straightforward to get
\begin{align*}
 & \mathrm{d}^2\SS|_{(\psi_0,r_0)}\big((\Delta \psi^{(1)},\Delta{r}^{(1)}),(\Delta \psi^{(2)},\Delta{r}^{(2)})\big)=\\& \Big( \mathrm{d}\TT_{\Delta \psi^{(1)}}|_{{{r}_{0,k}}}( \Delta {r}_k^{(2)})+\mathrm{d}\TT_{\Delta \psi^{(2)}}|_{{{r}_{0,k}}}( \Delta {r}_k^{(1)})+ \mathrm{d}^2{\TT_{\psi_0}}|_{{{r}_0}}(\Delta{r}_k^{(1)},\Delta{r}_k^{(2)})\Big)_{k=1}^K.
\end{align*}

Finally, we want to compose these expressions to get the gradient and bilinear Hessian for the function $f$ in \eqref{eqposcor}, which now can be written $f(p,\psi,r)=F\circ \LL(p,\SS\big(\psi,r)\big)$. 
Similar to the calculations in Section \ref{man}, we now introduce $a_0=I_K(p_0)$ and $b_0=\SS(\psi_0,r_0)$ and set $\Delta a=I_K(\Delta p)$,  but in contrast we put $$\Delta b=\mathrm{d}\SS|_{(\psi_0,r_0)}(\Delta \psi, \Delta r)+\frac{1}{2}\mathrm{d}^2\SS|_{(\psi_0,r_0)}\big[(\Delta \psi, \Delta r)^{\times 2}\big]$$ in place of $S_r(\Delta\psi)$.
To shorten formulas we also introduce the abbreviation $(\Delta a, \Delta b)^{\times 2}$ for $\big((\Delta a, \Delta b),(\Delta a, \Delta b)\big)$. Following the computations in Section \ref{man} it is easy to see that $F\circ \LL$ has the second order Taylor expansion \begin{align*}
&(F\circ \LL)(a_0+\Delta a, b_0+\Delta b)=F\circ\LL(a_0,b_0)+\Re \Big\langle \nabla F|_{\LL(a_0,b_0)}, \mathrm{d}\LL|_{(a_0,b_0)}(\Delta a, \Delta b)\Big\rangle+\\
&\frac{1}{2}\Re\Big\langle \nabla F|_{\LL(a_0,b_0)}, \mathrm{d}^2\!\LL|_{\!(a_0,b_0)}\big[\!(\Delta a, \Delta b)^{\times 2}\big]\Big\rangle \!+\! \frac{1}{2}\HH^{F}\!|_{\!\LL(a_0,b_0)}\Big[\big(\mathrm{d}\LL|_{(a_0,b_0)}(\Delta a, \Delta b)\big)^{\times 2}\Big]\!+\!\mathcal{O}\left(\|(\Delta a,\Delta b)\|^3\right)\\
&=f(p_0,\psi_0,r_0)+\Re \Big\langle \Phi_0, \Delta a\cdot b_0+ a_0\cdot \Delta b\Big\rangle+ \frac{1}{2}\HH^{F\circ\LL}|_{(a_0,b_0)}\Big[(\Delta a, \Delta b)^{\times 2}\Big]+\mathcal{O}\left(\|(\Delta a,\Delta b)\|^3\right)
\end{align*}
where $\HH^F$ is given in \eqref{HbiFG} and $\HH^{F\circ\LL}|_{(a_0,b_0)}$ is defined via  
\begin{multline*}
\HH^{F\circ\LL}|_{(a_0,b_0)}\Big[(\Delta a, \Delta b)^{\times 2}\Big]=\\\Re\Big\langle \nabla F|_{\LL(a_0,b_0)}, \mathrm{d}^2\LL|_{(a_0,b_0)}\big[(\Delta a, \Delta b)^{\times 2} \big]\Big\rangle + \HH^F|_{\LL(a_0,b_0)}\Big[\big(\mathrm{d}\LL|_{(a_0,b_0)}(\Delta a, \Delta b)\big)^{\times 2}\Big].
\end{multline*}
Upon inserting the concrete formulas for $\Delta a$ and $\Delta b$ we then see that the second order Taylor expansion of $f(p_0+\Delta p,\psi_0+\Delta \psi,r_0+\Delta r)=F\circ\LL\big(I_K(p_0+\Delta p),\SS(\psi_0+\Delta \psi,r_0+\Delta r)\big)$ becomes
\begin{align*}
    &f(p_0,\psi_0,r_0)+\Re \Big\langle \Phi_0, I_K(\Delta p)\cdot b_0+ a_0\cdot \Big(\mathrm{d}\SS|_{(\psi_0,r_0)}(\Delta \psi, \Delta r)+\frac{1}{2}\mathrm{d}^2\SS|_{(\psi_0,r_0)}\big[(\Delta \psi, \Delta r)^{\times 2}\big]\Big)\Big\rangle\\&\qquad + \frac{1}{2}\HH^{F\circ\LL}|_{(a_0,b_0)}\Big[ \big(I_K(\Delta p), \mathrm{d}\SS|_{(\psi_0,r_0)}(\Delta \psi, \Delta r)\big)^{\times 2}\Big]+\mathcal{O}\big(\|(\Delta p,\Delta \psi,\Delta r)\|^3\big)
\end{align*}
where we omitted the $\mathrm{d}^2\SS$-term from the second row since this part anyways gets absorbed by the ordo (as in the proof of Theorem \ref{t1}).

From the above expression we can now easily identify the gradient and the bilinear Hessian. Obviously, the gradients $\nabla_p f$ and $\nabla_\psi f$ are the same as in \eqref{nablaab} in Section \ref{man}, so we will not derive them again. Recalling \eqref{l04} and the expression for $\mathrm{d}\SS$, we have that the term that depends linearly on $\Delta r$ is 
\begin{align*}
&\Re \Big\langle \Phi_0, a_0\cdot\Big(  \mathrm{d}\TT_{\psi_0}|_{{{r}_{0,k}}}( \Delta {r}_k) \Big)_{k=1}^K\Big\rangle=\Re \Big\langle \Phi_0, I_K(p_0)\cdot\Big(  \mathrm{d}\TT_{\psi_0}|_{{{r}_{0,k}}}( \Delta {r}_k) \Big)_{k=1}^K\Big\rangle=\\& \Re \Big\langle \Phi_0, - 2\pi i \left( p_0\cdot \CC\Big[\mathcal{F}^{-1}\big(\langle\xi,  \Delta r\rangle e^{ -2 \pi i \langle {r}_{0,k}, {\xi}\rangle}\hat{\psi}_0(\xi)\big)\Big] \right)_{k=1}^K\Big\rangle    
\end{align*} 
where the operator $I_K$ disappears when we move $p_0$ inside the main parenthesis. From this expression it follows that $\nabla_{{r}} f|_{(p_0,\psi_0,r_0)} $ equals \begin{align*}
                              -2\pi \mathsf{Im} \Big(&\Big[ \Big\langle \Phi_{0,k}, p_0\cdot  \CC\Big(\mathcal{F}^{-1}\big[\xi_1 e^{ -2\pi i \langle{r}_{0,k},  {\xi}\rangle}\hat{\psi}_0(\xi)\big])\Big\rangle,\\&\Big\langle \Phi_{0,k},p_0\cdot \CC\Big(\mathcal{F}^{-1}\big[\xi_2 e^{ -2\pi i \langle{r}_{0,k},{\xi}\rangle}\hat{\psi}_0(\xi)\big]\Big) \Big\rangle\Big]\Big)_{k=1}^K,
                           \end{align*}
where we use $\Phi_{0,k}$ to denote the $k$-th slice of $\Phi_0$. Note that the only difference between the formula for the first and second coordinate is the swapping of $\xi_1$ for $\xi_2$.
Turning finally to the Hessian (on the ``diagonal'') we have 
\begin{equation*}\begin{aligned}
&\HH^f|_{ (p_0,\psi_0,r_0)}\big[(\Delta p,\Delta \psi,\Delta {r})^{\times 2}\big]=\\&\Re \Big\langle \Phi_0, \mathrm{d}^2\SS|_{(\psi_0,r_0)}\big[(\Delta \psi, \Delta r)^{\times 2}\big]\Big\rangle+ \HH^{F\circ\LL}|_{(a_0,b_0)}\Big[ \big(I_K(\Delta p), \mathrm{d}\SS|_{(\psi_0,r_0)}(\Delta \psi, \Delta r)\big)^{\times 2}\Big],
\end{aligned}
\end{equation*}
and the corresponding expression for the bilinear form immediately follows.
The expressions for the Hessian operator can be obtained following the computations outlined in Section \ref{sec:HesOp}, and are omitted here for sake of space.

\section*{Appendix II: Preconditioning for numerical stability}\label{prec}

When creating a large functional based on multiple variables, the variability of the functional can be much greater with respect to one variable, leading to numerical instabilities. 
To balance the situation, we propose rescaling the variables. Consider a functional $f$ of two high-dimensional variables $(x,y)$. We want to re-scale them as $\tilde x=x/\rho_x$ and $\tilde y=y/\rho_y$, so the new functional becomes $$\tilde{f}(\tilde{x},\tilde{y})=f(\rho_x \tilde x,\rho_y \tilde y).$$
We have that
\begin{align*}
 & \tilde f(\tilde{x}+\Delta \tilde{x},\tilde{{y}}+\Delta\tilde{{y}})=f(\rho_x\tilde{x}+\rho_x\Delta \tilde{x},\rho_y\tilde{{y}}+\rho_y\Delta\tilde{{y}})\\
 &\quad =f(\rho_x\tilde{x},\rho_y\tilde{{y}})+
 \mathsf{Re}\langle \nabla_x f|_{(\rho_x\tilde {x},\rho_y\tilde{{y}})},\rho_x \Delta \tilde x\rangle+\mathsf{Re}\langle \nabla_y f|_{(\rho_x\tilde{x},\rho_y\tilde{{y}})},\rho_y \Delta \tilde y\rangle\\
 &\qquad\qquad +\frac{1}{2}\HH^f|_{(\rho_x\tilde{{x}},\rho_y \tilde y)}\big[(\rho_x\Delta \tilde{x},\rho_y\Delta \tilde{{y}})^{\times 2}\big]+\mathcal{O}(\|(\Delta\tilde x,\Delta \tilde y)\|^3)
\end{align*}
from which it follows that 
\begin{align*}
  &\nabla_{\tilde{x}} \tilde f|_{(\tilde{x},\tilde{{y}})}=\rho_x\nabla_{x} f|_{(x,y)}, \qquad\nabla_{\tilde{y}} \tilde f|_{(\tilde{x},\tilde{{y}})}=\rho_y\nabla_{y} f|_{(x,y)}, \\
  &\HH^{\tilde{f}}|_{(\tilde{x},\tilde{{y}})}\big[(\Delta\tilde{x},\Delta\tilde{{y}})^{\times 2}\big]  =  \HH^{{f}}|_{({x},{{y}})}\big[(\rho_x \Delta\tilde{x},\rho_y \Delta\tilde{{y}})^{\times 2}\big], \\
  &H^{\tilde{f}}|_{(\tilde{x},\tilde{{y}})}(\Delta\tilde{x}, \Delta\tilde{{y}})=\mathrm{diag}_{(\rho_x,\rho_y)}H^f|_{({x},{{y}})}(\rho_x \Delta\tilde{x},\rho_y \Delta\tilde{{y}}),
\end{align*}
where $\mathrm{diag}_{(\rho_x,\rho_y)}(u,v)=(\rho_x u, \rho_y v)$.

To avoid introducing new variables in practice (i.e.~in the computer code), let us now work out the steps for optimizing $\tilde f$, translated to the original coordinates, starting with BH-CG. Given a point $(\tilde x^{(j)},\tilde y^{(j)})$ and old search directions $\tilde \eta^{(j-1)}= (\tilde \eta_{\tilde{x}}^{(j-1)},~\tilde \eta_{\tilde{y}}^{(j-1)})$, new search directions 
are given by 
$$\tilde \eta^{(j)}_{\tilde{z}}=-\nabla_{\tilde{z}} \tilde f|_{(\tilde{x}^{(j)},\tilde{{y}}^{(j)})}+\frac{\HH^{\tilde f}|_{(\tilde{x}^{(j)},\tilde{{y}}^{(j)})}(\nabla \tilde f|_{(\tilde{x}^{(j)},\tilde{{y}}^{(j)})},\tilde \eta^{(j-1)})}{\HH^{\tilde f}|_{(\tilde{x}^{(j)},\tilde{{y}}^{(j)})}( \tilde \eta^{(j-1)},\tilde \eta^{(j-1)})} \tilde \eta^{(j-1)}_{\tilde{z}},$$  where $z$ stands for either $x$ or $y$.
The updates for new variables are calculated as
\begin{equation}
\tilde{x}^{(j+1)}=\tilde{x}^{(j)}-\tilde{\alpha}^{(j)}\tilde{\eta}_x^{(j)},\qquad
\tilde{y}^{(j+1)}=\tilde{y}^{(j)}-\tilde{\alpha}^{(j)}\tilde{\eta}_y^{(j)}
\end{equation}
To work instead with original variables we multiply both sides by $\rho$,
\begin{equation}
x^{(j+1)}={x}^{(j)}-\rho_x\tilde{\alpha}^{(j)}\tilde{\eta}_x^{(j)},\qquad y^{(j+1)}={y}^{(j)}-\rho_y\tilde{\alpha}^{(j)}\tilde{\eta}_y^{(j)}
\end{equation}
and introduce $\eta_z^{(j)}$ by setting $\tilde \eta_{\tilde{{z}}}^{(j)} = \eta_z^{(j)}/\rho_z$ so that it scales just as $x$ and $y$. 
Then
\begin{equation}
\begin{aligned}
 &\eta^{(j)}_z=-\rho_z^2\nabla_z f|_{({x}^{(j)},{{y}}^{(j)})}\\ &+\frac{\HH^{ f}|_{({x}^{(j)},{{y}}^{(j)})}\Big[(\nabla_{{{{x}}}} \rho_x^2 f|_{({x}^{(j)},{{y}}^{(j)})},\rho_y^2\nabla_{{{{y}}}} f|_{({x}^{(j)},{{y}}^{(j)})}), ( \eta^{(j-1)}_{{{x}}},\eta^{(j-1)}_{{{y}}})\Big]}{\HH^{ f}|_{({x}^{(j)},{{y}}^{(j)})}\Big[(\eta^{(j-1)}_{{{x}}}, \eta^{(j-1)}_{{{y}}}),(\eta^{(j-1)}_{{{x}}}, \eta^{(j-1)}_{{{y}}})\Big]} \eta^{(j-1)}_{{{z}}}.  
 \end{aligned}
\end{equation}
Armed with this, we now want to minimize $$\alpha \mapsto \tilde f(\tilde x^{(j)},\tilde y^{(j)})-\alpha\Re\left\langle  \nabla \tilde f|_{(\tilde x^{(j)},\tilde y^{(j)})},\tilde \eta^{(j)}\right\rangle +\frac{\alpha^2}{2}\HH^{\tilde f}|_{(\tilde{x}^{(j)},\tilde{{y}}^{(j)})}\big[ \tilde \eta^{(j)},\tilde \eta^{(j)}\big]$$
which analogously yields \begin{equation}\label{sl}
\begin{aligned}
    \alpha^{(j)}=\frac{\Big\langle(\nabla_{{{{x}}}} f|_{({x}^{(j)},{{y}}^{(j)})},\nabla_{{{{y}}}} f|_{({x}^{(j)},{{y}}^{(j)})}), ( \eta^{(j)}_{{{x}}}, \eta^{(j)}_{{{y}}})\Big\rangle}{\HH^{ f}|_{({x}^{(j)},{{y}}^{(j)})}\Big[( \eta^{(j)}_{{{x}}}, \eta^{(j)}_{{{y}}}),( \eta^{(j)}_{{{x}}}, \eta^{(j)}_{{{y}}})\Big]}.
    \end{aligned}
\end{equation} 
We can now update the old variables according to $x^{(j+1)}=x^{(j)}-\alpha^{(j)} \eta_x^{(j)}$ and analogously for $y$, i.e.~we retrieve \eqref{eqalpha} unchanged.

For BH-GD we analogously apply the above formula for the steplength, setting $\eta_z^{(j)}=-\rho_z^2\nabla_{z} f|_{(x^{(j)},y^{(j)})}$, whereas for Quasi-Newton (BH-QN), we get the new search direction by approximately solving the equation system $H^{\tilde{f}}|_{(\tilde{x}^{(j)},\tilde{y}^{(j)})}\left(\tilde{\eta}^{(j)}_{\tilde{x}},\tilde{\eta}^{(j)}_{\tilde{y}}\right)=-\nabla\tilde{f}|_{(\tilde{x}^{(j)},\tilde{y}^{(j)})}$
which in the original coordinates becomes 
$$\mathrm{diag}_{(\rho_x^2,\rho_y^2)}H^{{f}}|_{({x}^{(j)},{y}^{(j)})}\left({\eta}^{(j)}_{{x}},{\eta}^{(j)}_{{y}}\right)=-\Big(\rho_x^2 \nabla_{{x}}{f}|_{({x}^{(j)},{y}^{(j)})},\rho_y^2 \nabla_{{y}}{f}|_{({x}^{(j)},{y}^{(j)})}\Big).$$

\bibliography{refs}

\end{document}